\documentclass[11pt]{article}

\usepackage{amssymb,
  amsmath,
  amsfonts,
  amsthm}
\usepackage[utf8]{inputenc}

\usepackage[plainpages=false]{hyperref}

\usepackage{booktabs}
\usepackage{mathtools}
\usepackage{authblk}
\usepackage{fullpage}
\usepackage{graphicx}
\usepackage{xcolor}

\newtheorem{definition}{Definition}
\newtheorem{theorem}{Theorem}
\newtheorem{lemma}[theorem]{Lemma}
\newtheorem{corollary}[theorem]{Corollary}

\usepackage{enumitem}

\newlist{conditions}{enumerate}{2}
\setlist[conditions,1]{label=\emph{(\arabic*)}, ref={\arabic*}}
\setlist[conditions,2]{label=\emph{(\arabic*)}, ref={\arabic{conditionsi}.\roman*}}

\renewcommand{\phi}{\varphi}

\newcommand{\fin}{\textup{Fin}}
\newcommand{\tower}{\textup{tower}}
\newcommand{\conn}{\text{conn}}
\newcommand{\EXP}[1][d]{#1\textnormal{-}\!\operatorname{exp}}

\makeatletter
\let\@@pmod\pmod
\DeclareRobustCommand{\pmod}{\@ifstar\@pmods\@@pmod}
\def\@pmods#1{\mkern4mu({\operator@font mod}\mkern 6mu#1)}
\makeatother

\newcommand{\logic}[1]{\textsc{#1}}
\newcommand{\logl}{\logic{l}}
\newcommand{\FO}{\logic{fo}}
\newcommand{\FOmod}{\logic{fo+mod}}
\newcommand{\MSO}{\logic{mso}}
\newcommand{\CMSO}{\logic{cmso}}

\newcommand{\ordinv}[1]{\textnormal{{\small \textless}-inv-}#1}
\newcommand{\oiFO}{\ordinv{\FO{}}}
\newcommand{\oiMSO}{\ordinv{\MSO{}}}

\newcommand{\ar}{\operatorname{ar}}

\newcommand{\eleq}[1][]{\equiv_{#1}}
\newcommand{\foeleq}[1][]{\eleq[#1]^{\FO{}}}
\newcommand{\msoeleq}[1][]{\eleq[#1]^{\MSO{}}}
\newcommand{\leleq}[1][]{\eleq[#1]^{\logl}}

\newcommand{\types}[1][\sigma,q,d]{\mathcal{T}_{#1}}
\newcommand{\ctypes}[1][\sigma,q,d]{\mathcal{T}^{\conn}_{#1}}
\newcommand{\atleq}{\preceq_{\text{atomic}}}

\newcommand{\maxeq}[1]{\equiv_{{\wedge}#1}}
\newcommand{\redmax}[2]{[#1]_{{\wedge}#2}}
\newcommand{\modeq}[1]{\equiv_{\mathrm{mod}\,#1}}
\newcommand{\redmod}[2]{[#1]_{\mathrm{mod}\,#2}}

\newcommand{\size}[1]{\|#1\|}
\newcommand{\qr}[1]{\operatorname{qr}(#1)}
\newcommand{\qad}[1]{\operatorname{qad}(#1)}

\newcommand{\st}{\mathbin |}
\newcommand{\absval}[1]{\vert #1 \vert}

\newcommand{\limplies}{\rightarrow}

\newcommand{\biglor}{\bigvee}
\newcommand{\bigland}{\bigwedge}

\newcommand{\leqsym}{\logic{\ensuremath{\leq}}}

\newcommand{\I}{\mathcal{I}}

\newcommand{\rela}[2]{\ensuremath{{#1|}_{#2}}}
\newcommand{\relao}[2]{\rela{#1\!}{#2}}

\DeclareMathOperator{\tp}{tp}
\DeclareMathOperator{\rtp}{rtp}

\DeclareMathOperator{\reach}{reach}

\newcommand{\class}[1]{\mathcal{#1}}
\newcommand{\struct}[1]{\mathfrak{#1}}

\newcommand{\AS}{\struct{A}}
\newcommand{\BS}{\struct{B}}
\newcommand{\CS}{\struct{C}}
\newcommand{\FS}{\struct{F}}
\newcommand{\HS}{\struct{H}}
\newcommand{\KS}{\struct{K}}
\newcommand{\TS}{\struct{T}}

\newcommand{\GG}{\struct{G}}
\newcommand{\iso}{\cong}

\newcommand{\td}{\operatorname{td}}

\DeclareMathOperator{\tdroot}{roots}
\DeclareMathOperator{\dist}{dist}

\newcommand{\num}{\operatorname{num}}
\newcommand{\height}{\operatorname{height}}

\newcommand{\set}[1]{\{#1\}}
\newcommand{\setc}[2]{\{#1 \st #2\}}
\newcommand{\lrsetc}[2]{\left\{#1 \st #2\right\}}

\newcommand{\N}{\mathbb{N}}
\newcommand{\Npos}{\mathbb{N}^{+}}

\newcommand{\intersect}{\cap}
\newcommand{\union}{\cup}
\newcommand{\bigunion}{\bigcup}
\newcommand{\disunion}{\sqcup}

\newcommand{\bigo}{O}

\newcommand{\nexp}[1][d]{\ensuremath{#1\textnormal{-}\textsc{exp}}}

\newcommand{\IONE}{\mathbf{1}}
\newcommand{\ITWO}{\mathbf{2}}
\newcommand{\IONEL}{\IONE_L}
\newcommand{\IONER}{\IONE_R}
\newcommand{\ITWOL}{\ITWO_L}
\newcommand{\ITWOR}{\ITWO_R}

\newcommand{\enc}{\operatorname{enc}} \newcommand{\inc}{\textsc{inc}}
\newcommand{\dec}{\textsc{dec}} \newcommand{\halt}{\textsc{halt}}

\begin{document}

\title{Succinctness of Order-Invariant Logics on\\ Depth-Bounded
  Structures\footnote{A  preliminary version of this paper was presented at the
    \textsc{mfcs} 2014 conference~\cite{EickmeyerEH2014}.}} 

\author[1]{Kord Eickmeyer}
\affil[1]{TU Darmstadt, Germany}
\author[2]{Michael Elberfeld}
\affil[2]{RWTH Aachen University, Germany}
\author[3]{Frederik Harwath}
\affil[3]{Goethe University Frankfurt, Germany}

\maketitle

\begin{abstract}
  We study the expressive power and succinctness of order-invariant
  sentences of first-order (FO) and monadic second-order (MSO) logic
  on structures of bounded tree-depth. Order-invariance is undecidable
  in general and, thus, one strives for logics with a decidable syntax
  that have the same expressive power as order-invariant sentences.
  We show that on structures of bounded tree-depth, order-invariant FO
  has the same expressive power as FO. Our proof technique allows for
  a fine-grained analysis of the succinctness of this translation. We
  show that for every order-invariant FO sentence there exists an FO
  sentence whose size is elementary in the size of the original
  sentence, and whose number of quantifier alternations is linear in
  the tree-depth.  We obtain similar results for MSO. It is known that
  the expressive power of MSO and FO coincide on structures of bounded
  tree-depth.  We provide a translation from MSO to FO and we show
  that this translation is essentially optimal regarding the formula
  size.  As a further result, we show that order-invariant MSO has the
  same expressive power as FO with modulo-counting quantifiers on
  bounded tree-depth structures.
\end{abstract}

\section{Introduction}
\label{sec:introduction}

Understanding the \emph{expressivity} of logics on finite structures---the
question of which properties are definable in a certain logic---plays an
important role in database and complexity theory. In the former, logics are used
to formulate queries; in the latter, they describe computational
problems. Moreover, besides just studying a logic's expressivity, understanding
its \emph{succinctness}---the question of how complex definitions of properties
such as queries and problems must be---is a requirement towards (theoretical)
expressivity results of (potential) practical importance. The present work
studies the succinctness of first-order logic (\FO{}) as well as its
succinctness compared to extensions allowing for the use of a linear order and
set quantifiers. This extends and refines recent studies on the expressivity of
these logics~\cite{BenediktSegoufin2009,ElberfeldGT12} on restricted classes of
structures. The structures we consider have bounded tree-depth, which is a graph
invariant that measures how far a graph is from being a star in a similar way as
tree-width measures how far a graph is from being a tree. Our results are
summarised by Figure~\ref{fig:results}.

\begin{figure}[t]
  \begin{center}
    \begin{tabular}{c @{~~~~~~~} c @{~~~~~~~} c @{~~~~~~~} c}
      \toprule
      $\phi \in$
      & \oiFO{} & \MSO{} & \oiMSO
      \\
      \midrule
      $\psi \in$ & \FO{} & \FO{} & \FOmod
      \\
      $\size{\psi}$ & $\nexp[d](q)$ &
      $\nexp[d](q)$ & non-elementary
      \\
      $\qad{\psi}$ & $\bigo(d)$ & $\bigo(d)$ & $\bigo(d)$\\
      \bottomrule
    \end{tabular}
    \end{center}
  \caption{Summary of our results: A formula $\phi$ of quantifier
    rank $q$ is translated into a formula $\psi$ that is equivalent
    to $\phi$ on structures of tree-depth at most $d$.}
  \label{fig:results}
\end{figure}

In both database and complexity theory, one often assumes that structures come
with a linear order and formulae are allowed to use this order as long as the
properties defined by them do not depend on the concrete interpretation of the
order in a structure. Such formulae are called \emph{order-invariant}. Since
testing order-invariance for given \FO{}-formulae is undecidable in general, one
tries to find logics that have the same expressive power as order-invariant
formulae, but a decidable syntax. Several examples prove that order-invariant
\FO{}-formulae (\oiFO{}) are more expressive than \FO{}-formulae without access
to orders, cf.~\cite{Schweikardt2013}.  A common feature of these separating
examples is that their Gaifman graphs contain large cliques, making them rather
complicated from the point of view of graph structure theory. 

For tree structures, on the other hand, \cite{BenediktSegoufin2009} showed that
the expressivity of \FO{} and \oiFO{} coincide. Following this example, 
we show that on structures of tree-depth at most $d$ each \oiFO{}-sentence can be translated to an \FO{}-sentence
whose size is $d$-fold exponential in the size of the original sentence (Theorem~\ref{thm:oifo-eq-fo}). The importance of the expressivity result is
highlighted by the fact that order-invariance is undecidable even on structures
of tree-depth at most $2$ (Theorem~\ref{thm:ordinv-undecidable}). 

A logic that is commonly studied from the perspectives of algorithm design and
language theory is monadic second-order logic (\MSO{}), which extends
\FO{}-formulae by the ability to quantify over sets of elements instead of just
single elements. While it has a rich expressivity that exceeds that of \FO{}
already on word structures, the expressive powers of \FO{} and \MSO{} coincide
on any class of structures whose tree-depth is bounded~\cite{ElberfeldGT12} by a constant $d$. We
refine this by presenting a translation into \FO{}-formulae of $d$-fold exponential size
(Theorem~\ref{thm:mso-eq-fo}). We prove that this translation is essentially optimal regarding the formula size
(Theorem~\ref{thm:mso-lower-bound}).
Beside the succinctness results, we prove that \oiMSO{} has the same
expressive power as \FOmod{}, the extension of \FO{} by arbitrary
first-order modulo-counting quantifiers, for structures of bounded tree-depth (Theorem~\ref{thm:oimso-eq-fomod}). 

Our results also have implications for \FO{} itself. They imply that the
quantifier alternation hierarchy for \FO{} of \cite{ChandraH1982} collapses on
structures of bounded tree-depth, whereas it is shown in~\cite{ChandraH1982} to
be strict on trees of unbounded height. For structures of bounded tree-depth we
are able to turn any \FO-formula into a formula whose size is bounded by the
quantifier depth of the original formula and whose quantifier alternation depth
is bounded by a linear function in the tree-depth.

A recurring theme in the study of \FO{}, \MSO{}, and their variants is the
question of which graph-theoretical properties can be defined using formulae of
these logics. The main motivation behind these questions lies in the fact that
access to certain tree-decompositions or embeddings of the structure can be used
as a proof ingredient for translating formulae. Independent of the results
stated above, we prove that, for structures of bounded tree-depth, it is
possible to define tree-decompositions of bounded width \emph{and height} in
\FO{} (cf. Section~\ref{sec:canondecomp}).

\paragraph{Proof techniques}

Our proofs use techniques from finite model theory, in particular interpretation
arguments, logical types, and games. Compared to prior works
like~\cite{ElberfeldGT12}, we enrich the application of these techniques by a
quantitative analysis, thereby obtaining succinct translations instead of just
equal expressivity results. The proofs of \cite{ElberfeldGT12} use an involved
constructive variant of the Fefer\-man--Vaught composition theorem, which
complicates a straightforward analysis of the formula size in the translation
from \MSO{} to \FO{}. We also use composition arguments, but we get along with
an easier non-constructive variant. There is another proof of the result of
\cite{ElberfeldGT12} in \cite{GajarskyH2012}, but it relies on involved
combinatorial insights that seem unsuited for both a tight analysis of
succinctness as well as an adaptation to the ordered setting.

The results of \cite{BenediktSegoufin2009} about the expressivity of
\oiFO{} on trees use automata-theoretic and algebraic methods. Since
these methods seem unsuited to obtain succinct formula translations,
we apply and develop techniques that are mainly based on using games:
In order to translate \oiFO-sentences into \FO-sentences, we first
restrict our attention to a certain kind of linear ordering that is
based on the logical types of recursively-defined substructures. Since
the \FO{}-type of ordered structures turns out to be \FO-definable
in the original (unordered) structures,
we are able to prove a
succinct translation from \oiFO{} to \FO.

In order to translate \oiMSO-sentences into \FOmod-sentences, the proof
structure is similar, but we need to add a ``pumping lemma'' for \oiMSO, which
proves the limited expressive power of \oiMSO{} on the recursively considered
substructures.

\paragraph{Organisation of this paper} The paper continues with a background
section and, then, the results related to \oiFO{}, \MSO{}, and \oiMSO{} are
proved in Sections~\ref{sec:oifo}, \ref{sec:mso}, and \ref{sec:oimso},
respectively. Tree-decompositions for structures of bounded tree-depth are
handled in Section~\ref{sec:canondecomp}.

\section{Background}
\label{sec:background}

In the present section, we review definitions and terms related to logical
formulae and structures as well as the notion of tree-depth.

\paragraph{General notation}

The sets of natural numbers with and without $0$ are denoted, respectively, by
$\N$ and $\Npos$.  Let $[i,j]:=\set{i, \ldots, j}$ for all $i,j\in \N$ with
$i\leq j$, and let $[j]:=[1,j]$. We define the \emph{$d$-fold exponential
function} $\EXP[d](n)$ recursively by $\EXP[0](n) := n$, and $\EXP[(d+1)](n) :=
2^{\EXP[d](n)}$.  The \emph{class of functions that grow at most $d$-fold
exponentially} is $\nexp[d] := \{ f:\N \to \N \mathbin{\st} f(n) \leq
\EXP[d](n^c)\text{ for some }c \in \N\text{ and all }n > c \}$. If we say that a
relation is an \emph{order}, we implicitly assume that it is linear. Thus an
order is an antisymmetric, transitive, reflexive and total binary relation.

\paragraph{Logic}

For a reference on notation and standard methods in finite model theory, we
refer to the book of \cite{Libkin2004}.  We denote structures by Fraktur
letters $\AS, \BS, \CS, \ldots$ and their universes by the
corresponding latin letters $A, B, C, \ldots$.  Besides the standard logics
$\FO{}$ and $\MSO{}$, we also consider the logic $\FOmod{}$ that is obtained
from $\FO{}$ by allowing the use of \emph{modulo-counting quantifiers}
$\exists^{i \pmod*{p}}$ for each $i\in \N$ and $p\in \Npos$.  The meaning of
these quantifiers is that $\AS \models \exists^{i \pmod*{p}} x\,
\phi(x,\bar{y})$ iff $|\set{b \in A : \AS \models \phi(b,\bar{a})}| \equiv i
\pmod*{p}$, where $\AS$ is a structure and $\bar{a}$ is a tuple of its elements.

We write $\qr{\phi}$ for the \emph{quantifier rank} and $\size{\phi}$ for the
\emph{size} (or \emph{length}) of a formula~$\phi$.  The \emph{quantifier
alternation depth} $\qad{\phi}$ of a formula $\phi$ in \emph{negation normal
form} (\textsc{nnf}, i.e. all negations of $\phi$ occur directly in front of
atomic formulae) is the maximum number of alternations between $\exists$- and
$\forall$-quanti\-fiers on all directed paths in the syntax tree of $\phi$. If
$\phi$ is not in \textsc{nnf}, we first find an equivalent formula $\phi'$ in
\textsc{nnf} using a fixed conversion procedure and then define $\qad{\phi} :=
\qad{\phi'}$.  If $\Phi$ is a set of formulae, we let $\size{\Phi} :=
\max_{\phi\in \Phi} \size{\phi}$ and $\qad{\Phi} := \max_{\phi\in \Phi}
\qad{\Phi}$.

For any logic $\logl\in \set{\FO{},\FOmod{}, \MSO{}}$ and $q\in \N$, we write
$\AS\leleq[q] \BS$ for $q \in \mathbb{N}$ if $\sigma$-structures $\AS$ and $\BS$
satisfy the same $\logl[\sigma]$-sentences of quantifier rank at most $q$. The
$\leleq[q]$-equivalence class of $\AS$ is its \emph{$(\logl,q)$-type} and is
denoted by $\tp_{\logl,q}(\AS)$.  If the logic $\logl$ has been fixed or the
concrete logic is not important for the discussion, we omit it in this and
similar notation.

For a signature $\sigma$, we denote by $\sigma^{\leq}$ the signature $\sigma
\cup \{{\leq}\}$, where ${\leq} \not\in \sigma$ is a binary relation symbol. An
\emph{ordered $\sigma^{\leq}$-structure} is a $\sigma^{\leq}$-structure $\AS$
where $\leqsym^\AS$ is an order on $\AS$. An \emph{ordered expansion}
$(\AS,\preceq)$ of a $\sigma$-structure $\AS$ is an expansion of $\AS$ to an
ordered $\sigma^{\leq}$-structure. A sentence $\phi \in \FO{}[\sigma^{\leq}]$ is
\emph{order-invariant} on a class $\class{C}$ of structures if for all
$\sigma$-structures $\AS \in \class{C}$ and all ordered expansions
$(\AS,\preceq_1)$ and $(\AS,\preceq_2)$ of $\AS$ we have $(\AS, \preceq_1)
\models \phi$ iff $(\AS, \preceq_2) \models \phi$.  If $\class{C}$ is not
otherwise stated, we assume $\class{C}$ to be the class of all finite
structures. The set of all order-invariant $\phi \in \FO{}[\sigma, \leqsym]$ is
denoted by $\oiFO{}[\sigma]$, and for such a $\phi$ and a $\sigma$-structure
$\AS$ we write $\AS \models_{\leqsym} \phi$ if $(\AS, \preceq) \models \phi$ for
some (equivalently, for every) ordered expansion $(\AS,\preceq)$ of $\AS$;
$\oiMSO{}$ is defined in the same way.

The \emph{restriction} of a binary relation $R$ on a set $M$ to a subset $N
\subseteq M$ is the relation $\rela{R}{N}:=\{(x,y) \in R : x,y \in N\}$. Note
that a substructure of an ordered structure is again an ordered structure.  For
two linear orders $\preceq_1$ and $\preceq_2$ on disjoint sets $M_1$ and $M_2$,
we define a linear order $\preceq_1 + \preceq_2$ on $M_1 \union M_2$, the
\emph{(ordered) sum} of $\preceq_1$ and $\preceq_2$, as $\preceq_1 \cup
\preceq_2 \cup\, (M_1 \times M_2)$.

If $\phi(\bar y)$ is a formula and $\psi(\bar x, z)$ is a formula with at least
one free variable $z$, then $\rela{\phi}{\psi}(\bar x,\bar y)$ is the
\emph{relativisation of $\phi$ to $\psi$}. We construct $\rela{\phi}{\psi}$ by
replacing subformulae $\exists x\, \chi$ and $\forall x\,\chi$ by $\exists
x\,(\psi(\bar y,x) \land \rela{\phi}{\psi})$ and $\forall z\,(\psi(\bar y, x)
\limplies \rela{\phi}{\psi})$, respectively. Note that $\qad{\phi|_\psi} =
\qad{\phi}$ if $\psi$ is an existential formula; in particular, $(\psi(\bar y,
x) \limplies \rela{\phi}{\psi}) \equiv (\lnot\psi(\bar y, x) \lor
\rela{\phi}{\psi})$ where, in this case, $\lnot\psi(\bar y,x)$ is equivalent to
a universal formula.

We transfer graph theoretic notions from graphs to general structures via the
notion of Gaifman graphs. The \emph{Gaifman graph} $\GG(\AS)$ of a structure
$\AS$ is the simple undirected graph with vertex set $A$ containing an edge
between $x,y\in A$ iff $x \neq y$ and $x$ and $y$ occur together in a tuple in
one of the relations of $\AS$. The \emph{distance} $\dist_{\AS}(a,b)$ between
elements $a,b$ of $\AS$ is their distance in $\GG(\AS)$, i.e. the length of a
shortest path between $a$ and $b$ in $\GG(\AS)$.  Similarly, notions such as
\emph{connectivity} and (connected) \emph{components} of $\AS$ are defined. Note
that the edge relation of the Gaifman graph is definable by an existential
formula $\varphi_E(x,y)$, and this can be used to obtain, for every $\ell \geq
0$, an existential formula $\dist_{\leq \ell}(x,y)$ such that $\AS \models
\dist_{\leq \ell}(a,b)$ iff $\dist_{\AS}(a,b) \leq \ell$.

\paragraph{Encoding information about elements in extended
  signatures}

In our proofs we will repeatedly remove single elements $r$ from
structures $\AS$ and encode information about the relations between
$r$ and the remaining elements into an expansion $\AS^{[r]}$ of the
structure $\AS \setminus r$ (which is the substructure of $\AS$
induced on the elements different from $r$). We do this in such a way
that the $q$-type of $\AS$ is determined by the $q$-type of
$\AS^{[r]}$ together with what we call the atomic type of $r$ in
$\AS$.

The \emph{atomic type} $\alpha(\AS,a)$ of an element $a$ of a $\sigma$-structure
$\AS$ is the set of all $R\in \sigma$ such that $(a, \ldots, a)\in R^{\AS}$
(where the tuple $(a, \ldots, a)$ has length $\ar(R)$). If no confusion seems
likely we omit $\AS$ and just write $\alpha(a)$. Thus an atomic type is a
subset of $\sigma$, and we identify $\alpha \subseteq \sigma$ with the
$\FO[\sigma]$-sentence
$\alpha(x)\ :=\ \smashoperator{\bigland_{R\in \alpha}} R(x) \ \land
\ \smashoperator{\bigland_{R\in \sigma\setminus\alpha}} \lnot
R(x)$.

Since we will often need the atomic type of the $\leq$-minimal element
of a structure, we denote by $\alpha_{\AS}$ the type $\alpha(r,\AS)$
if $\AS$ is an ordered structure with minimal element $r$.

To encode the relations between the element which is removed and the
remaining elements, we define a signature $\tilde\sigma$ which
contains, for each $R\in \sigma$ and each nonempty $I \subseteq
[1,\ar(R)]$, a relation symbol $R_I$ of arity $\absval{I}$. Given a
structure $\AS = (A,(R^{\AS})_{R \in \sigma})$ and an element $r \in
A$ we now obtain a $\tilde\sigma$-structure $\AS^{[r]} = (A,
(R_I^{\AS^{[r]}})_{R_I \in \tilde\sigma})$ by setting
\[
R_I^{\AS^{[r]}} := \setc{(a_i)_{i \in I}}{ (a_1,\ldots, a_{\ar(R)}) \in
R^{\AS}\text{ and }a_i = r\text{ for }i \not\in I }.
\]
Note that $R^{\AS} = R_{[1,\ar(R)]}^{\AS^{[r]}}$, so up to a renaming
of relation symbols, $\AS^{[r]}$ is an expansion of $\AS \setminus
r$.

The $(L,q)$-type of $\AS$ is determined by $\alpha(r)$ and the
$(L,q)$-type of $\AS^{[r]}$:
\begin{lemma}
  \label{lem:IndASrtoAS}
  Let $\logl\in \set{\FO{}, \MSO{}}$ and $q\in \Npos$. Let $\AS$ and $\BS$ be
  structures, $r \in A$ and $s \in B$. If
  \[
  \alpha(\AS,r) = \alpha(\BS,s)
  \quad\text{and}\quad
  \tp_{\logl,q}(\AS^{[r]}) = \tp_{\logl,q}(\BS^{[s]}),
  \]
  then also
  \[
  \tp_{\logl,q}(\AS) = \tp_{\logl,q}(\BS).
  \]
\end{lemma}
\begin{proof}
  The same argument works for $\logl=\FO{}$ and $\logl=\MSO{}$.
  Duplicator has a winning strategy $\mathcal{S}$ in the $q$-round
  Ehrenfeucht-Fra\"iss\'e game for $\logl$ on $\AS^{[r]}$ and
  $\BS^{[s]}$. Note that the strategy $\mathcal{S}$ is, in particular,
  a winning strategy on $\AS \setminus r$ and $\BS \setminus s$,
  because $\AS^{[r]}$ and $\BS^{[s]}$ are expansions of these
  structures. Duplicator can win the $q$-round EF-game on $\AS$ and
  $\BS$ if she plays according to $\mathcal{S}$ on $\AS \setminus r$
  and $\BS \setminus s$, and if she responds to $r$ with $s$ and vice
  versa.

  We have to argue that this strategy preserves relations between the
  played elements. For relations not involving the removed elements
  $r$ and $s$, this is true because $S$ is a winning strategy for the
  $q$-round game on $\AS \setminus r$ and $\BS \setminus s$. Relations
  involving only the minimal elements are preserved because
  $\alpha(\AS,r) = \alpha(\BS,s)$. Relations involving the minimal
  elements and other elements are preserved, because they are encoded
  in the relations $R_I$ of the extended signature $\tilde \sigma$,
  and these are preserved by $\mathcal S$.
\end{proof}

The following lemma is easy to prove following these definitions:
\begin{lemma}
  \label{lem:ASrtoAS}
  Let $\logl\in \set{\FO{}, \FOmod{}}$.  For every $\logl[\tilde
  \sigma]$-sentence $\phi$ there is an $\logl[\sigma]$-formula
  $\I(\varphi)(z)$ of the same quantifier rank and quantifier
  alternation depth such that
  \[
  \AS \models \I(\phi)(r) \quad\text{iff}\quad \AS^{[r]} \models \phi,
  \]
  for all $\sigma$-structures $\AS$ and $r \in A$.
\end{lemma}
\begin{proof}
  The proof uses a standard interpretation argument.  It suffices to
  provide quantifier-free formulae with a parameter $z$ which define
  the universe and the relations of $\AS^{[r]}$ in $\AS$, provided
  that $z$ is interpreted by the element $r$. The universe is defined
  by the formula $x \neq z$. Let $R_I\in \tilde\sigma$. If, for each $i\leq
  \ar(R)$, we let 
  \[ y_i :=
  \begin{cases}
    x_j & \text{ if } i=i_j\in I\\
    z & \text{ if } i\notin I\\
  \end{cases}
  \]
  then $R(y_1, \ldots, y_{\ar(R)})$ is a formula with free variables
  $z,x_1,\ldots,x_{|I|}$ which defines $R_I^{\AS^{[r]}}$ in $(\AS,r)$.
\end{proof}

\paragraph{Tree-depth}

The following inductive definition is one of several equivalent ways
to define the \emph{tree-depth} $\td(G)$ of a graph
(see \cite{NesetrilMendez2012} for a reference on tree-depth):
\[
\td(G)\ :=\ \begin{cases}
  1 & \text{ if } |V(G)| = 1\\
  1 + \min_{\,r\in V(G)} \td(G\setminus r) & \text{ if $G$ is connected and } |V(G)| > 1\\
  \max_{\,i\in [n]} \td(K_{i}) & \text{ if $G$ has components $K_{1}, \ldots, K_{n}$}.\\
\end{cases}
\]

As usual, the tree-depth $\td(\AS)$ of a relational structure $\AS$ is
defined by $\td(\AS) := \td(\GG(\AS))$. We let
\[
\fin^\conn_\sigma \ := \ \setc{\AS\in \fin_\sigma}{\AS \text{ is
    connected}}
\]
and for each $d\in \Npos$, we let
\[
\begin{split}
  \fin_{\sigma,d} &:= \setc{\AS \in \fin_\sigma}{\td(\AS) \leq d},
  \\
  \fin^\conn_{\sigma,d} &:= \setc{\AS \in \fin^\conn_\sigma}{\td(\AS) \leq d}.
\end{split}
\]

As an immediate consequence of the above definition of tree-depth, each $\AS\in
\fin^\conn_{\sigma,d}$ with $d > 1$ contains an element $r$ with $\td(\AS
\setminus r) \leq \td(\AS) - 1$. We call these vertices \emph{tree-depth roots}
and denote the set of all such vertices by $\tdroot(\AS)$. By a result of
\cite{BoulandDK12}, the size of $\tdroot(\AS)$ is bounded by a function of $d$
(independent of the size of $\AS$): 

\begin{lemma}[{\cite[Lem. 7]{BoulandDK12}}]
  \label{lem:roots}
  There is a function $f : \Npos \to \Npos$ such that
  $\absval{\tdroot(G)} \leq f(\td(G))$ for each connected graph $G$.
\end{lemma}

Note that the definition of $\tdroot(G)$ in \cite{BoulandDK12} is slightly
different from ours, but the two definitions are easily seen to be equivalent.

A graph of tree-depth at most $d$ can not contain a path of length
$2^d$ (cf.~\cite[6.2]{NesetrilMendez2012}). Therefore
$\dist_\AS(a,b) < 2^d$
for all elements $a$ and $b$ in the same connected component of a
structure $\AS$ of tree-depth at most $d$, and the formula
$\reach_{d}(x,y) := \dist_{\leq 2^{d}}(x,y)$
defines the reachability relation in these structures:
\[
\AS \models \reach_{d}[a,b]
\quad\text{iff}\quad
a\text{ and }b\text{ belong to the same component of }\AS.
\]
This (existential) formula allows us to relativise a formula
$\phi(x)$ to the connected component of $x$:
\[
\AS \models \rela{\phi}{\reach_d(x,z)}[a]
\quad\text{iff}\quad
\KS \models \phi[a],
\]
where $\KS$ is (the substructure of $\AS$ induced on) the connected component of
$a$ in $\AS$.  Since $\reach_d$ is existential, we have
$\qad{\rela{\phi}{\reach_d(x,z}}) = \qad{\phi}$ .

Using these observations and the inductive definition of tree-depth,
it is easy to write down an $\FO[\sigma]$-sentence that defines
$\fin_{\sigma,d}$ on the class of all finite
$\sigma$-structures. While this na\"ive approach leads to a formula
whose quantifier alternation depth grows linearly with $d$, it is also
possible to construct a \emph{universal} sentence $\td_{\leq d}$
defining $\fin_{\sigma,d}$ as a subclass of $\fin_{\sigma}$,
cf.~\cite[Section~6.10]{NesetrilMendez2012} for details.
Using this sentence, we construct a sentence that defines the set $\tdroot(\AS)$ for each $\AS\in \fin^\conn_{\sigma,d}$ with $d>1$. To this end, we let $\tdroot_d(x)\ :=\ \smashoperator{\biglor_{c \leq d-1}} \big( \td_{> c} \ \land \ \rela{{\td_{\leq c}}}{(x \neq z)}(x) \big)$. 

\section{Order-invariant first-order logic}
\label{sec:oifo}

It is well-known that order-invariance is undecidable on the class $\fin_\sigma$
of all finite $\sigma$-structures, i.e.  there is no algorithm which decides for
a given $\FO[\sigma^{\leq}]$-sentence if it is order-invariant on
$\fin_\sigma$. This leads to the question if the expressive power of
order-invariant sentences on a class $\class{C}$ can be captured by a logic with
a decidable syntax.  An answer to this question in the case of the class
$\fin_\sigma$ seems out of reach.  We consider the question in the case of
bounded tree-depth structures, i.e. $\class{C}=\fin_{\sigma,d}$ for some $d\in
\Npos$.  More concretely, our aim is a proof of the following theorem:

\begin{theorem}
  \label{thm:oifo-eq-fo}
  For every $d\in \Npos$, every signature $\sigma$, and each sentence
  $\phi$ of $\oiFO[\sigma]$, there is an $\FO[\sigma]$-sentence $\psi$
  which is equivalent to $\phi$ on $\fin_{\sigma,d}$ and which has
  size $\size{\psi}\in \nexp[d](\qr{\phi})$ and quantifier-alternation
  depth $\qad{\psi} \leq 3d$.
\end{theorem}

The proof of Theorem~\ref{thm:oifo-eq-fo} will be presented in
Section~\ref{sec:oifo-to-fo} below. Before that, we want to motivate
Theorem~\ref{thm:oifo-eq-fo} by showing that the undecidability of
order-invariance holds even for structures of tree-depth 2.

\subsection{Undecidability of order-invariance on structures of tree-depth 2}
\label{sec:undecidability}

As mentioned by \cite{Schweikardt2013}, order-invariance on $\fin_\sigma$ is
decidable if the signature $\sigma$ contains only unary relation symbols.  An
ordered $\sigma$-structure in which the unary relations partition the universe
can be regarded as a word. An $\FO[\sigma^{\leq}]$-sentence $\phi$ then defines
a language $L_\phi$. The sentence $\phi$ is order-invariant iff the syntactic
monoid of $L_\phi$ is commutative, which is decidable.  This argument can be
extended to general $\sigma$-structures and to structures of tree-depth $1$ over
arbitrary signatures.

Hence, order-invariance is decidable on $\fin_{\sigma,d}$ if $d=1$. The next
theorem shows that it becomes undecidable for $d\geq 2$.

\begin{theorem}
  \label{thm:ordinv-undecidable}
  There is a signature $\sigma$ such that order-invariance is undecidable on $\fin_{\sigma,2}$.
\end{theorem}

The proof of Theorem~\ref{thm:ordinv-undecidable} uses a reduction from the
undecidable halting problem for \emph{counter machines} (cf. \cite{Minsky1967})
with two counters which store natural numbers. A counter machine executes a
\emph{program}, i.e. a finite sequence of the following \emph{instructions}:
\begin{description}
\item[$\inc(i)$]  increment counter $i$, proceed with next instruction.
\item[$\dec(i,j_{0},j_{1})$] if counter $i$ is not zero: decrement counter $i$,
  proceed with $j_1$-th instruction otherwise: proceed with instruction $j_0$.
\item[$\halt{}$] stop the execution. 
\end{description}

The configuration of the machine at any execution step is fully described by
a triple $(n_1,n_2,j)$, where $n_1,n_2 \geq 0$ are natural numbers
stored in the counters and $j\geq 1$ is the number of the next
instruction to be executed. Without loss of generality, we assume
that the last instruction of a program is always the $\halt{}$
instruction and that this instruction occurs nowhere else in the
program. Hence we say that a program \emph{halts} (on empty input) if
it ever reaches its last instruction when run from the initial
configuration $(0,0,1)$. 

\begin{proof}[Proof of Theorem~\ref{thm:ordinv-undecidable}]
  We say that a sentence $\phi\in \FO[\sigma^{\leq}]$ is
  \emph{$d$-satisfiable} if it has a model $(\AS,\leqsym^\AS)$ where
  $\AS\in\fin_{\sigma,d}$.  The folklore proof which shows that
  order-invariance on $\fin_\sigma$ is undecidable uses a many-one
  reduction from the undecidable finite satisfiability problem to
  order-invariance.  The same kind of argument proves that
  \emph{$d$-satisfiability} (i.e. the problem which asks if a given
  sentence $\phi\in \FO[\sigma^{\leq}]$ is $d$-satisfiable) many-one
  reduces to order-invariance on $\fin_{\tilde\sigma,d}$, where $\tilde\sigma :=
  \sigma \union \set{P}$ for a unary relation symbol $P\notin \sigma$.
  This follows from the fact that $\phi\in\FO[\sigma^{\leq}]$ is
  $d$-satisfiable if, and only if, the $\FO[{\tilde\sigma}^{\leq}]$-sentence
  $\phi \land \exists x \forall y\ (x \leqsym y\ \land P(x))$
  is \emph{not} order-invariant on $\fin_{\tilde\sigma,d}$.

  Hence, to complete the proof of our theorem, it suffices to show
  that the $2$-satisfiability problem is undecidable for some
  signature $\sigma$ to be fixed below. To this end we reduce the
  halting problem for counter machines to $2$-satisfiability.  Let
  $P=I_1\dotsb I_\ell$ be a program. We construct an
  $\FO[\sigma^{\leq}]$-sentence $\phi$ which is
  $\fin_{\sigma,2}$-satisfiable iff $P$ halts.  First we fix an
  encoding of configurations of $P$ by words over a finite alphabet
  $\Sigma$. It would be natural to do this by encoding the counter values
  in unary using different symbols; say, $(2,3,1)$ would become
  $\IONE\IONE\,\ITWO\ITWO\ITWO\,1$. We change this representation
  slightly: a configuration $(n_1,n_2,j)$ of $P$ is encoded by a word
  \[
  \enc(n_1,n_2,j)\ :=\ (\IONEL\IONER)^{n_1}\, (\ITWOL\ITWOR)^{n_2}\, j
  \]
  over the alphabet $\Sigma:=\set{\IONEL,\ITWOL, \IONER,\ITWOR, 1,
    \ldots, \ell}$.
  \footnote{This alphabet depends on the length of
    the given program $P$, but the proof can be
    modified easily to make the alphabet $\Sigma$, and therefore the signature $\sigma$, independent of $P$ without increasing the tree-depth of the structures involved.}

  Let $\sigma := \set{E} \union \tau$ where $E$ is a binary relation
  symbol and $\tau := \setc{P_a}{a\in \Sigma}$, where the $P_a$ are
  unary relation symbols.  The $\sigma$-structures that we consider
  are \emph{$\Sigma$-coloured graphs}, i.e. $\sigma$-structures where
  $E$ is the edge relation of a simple undirected graph and where the
  unary predicates are a vertex colouring (i.e. a partition of the
  vertex set). If a vertex of such a graph belongs to a relation
  $P_a$, we say that it is \emph{$a$-coloured}. The class of
  $\Sigma$-coloured graphs is obviously $\FO$-definable on
  $\fin_\sigma$.

  As usual, we identify each non-empty word over the alphabet $\Sigma$
  with an ordered $\tau$-structure which, in turn, we regard as an
  ordered $\Sigma$-coloured graph with no edges.  We refer to vertices
  which are coloured by $1, \ldots, \ell$ as \emph{instruction
    vertices}. If our program $P$ halts after at most $h$ computation
  steps then, with respect to our encoding, there exists a unique word
  $w_P$ which encodes the \emph{run} of $P$, i.e. the finite sequence
  of configurations at time steps $1, \ldots, h$. We want to define a
  class of ordered $\Sigma$-coloured graphs of maximum degree $1$
  obtained from the edge-less graph $w_P$ by adding edges between its
  vertices. These graphs will be called \emph{matching extensions} of
  $w_P$, since their edge relations will be unions of matchings
  (i.e. edge relations of graphs where each vertex is incident to
  exactly one edge).  Consider any word $w=\enc(C_1)\dotsb\enc(C_k)$
  which encodes a sequence of representations. We phrase the
  description of the execution of the counter machine program $P$
  given in the definition of counter machines above somewhat more
  formally as conditions under which the sequence $C_1, \ldots, C_k$
  is a run of $P$ (i.e. $w=w_P$).  At the same time, we rephrase them
  as statements about the ordered $\Sigma$-coloured graph $w$ in a way
  that will be suitable for the definition of our sentence $\phi$.
  \begin{enumerate}
  \item $C_1=(0,0,1)$ and $C_k$ is a halting configuration,
    i.e. $C_k=(n_1,n_2,\ell)$ for some $n_1,n_2 \geq 0$.
  
    With our encoding, this is equivalent to the first vertex of $w$ being
    $1$-coloured and the last vertex being $\ell$-coloured. (Recall
    that the machine starts with both counters being $0$.)
  
  \item For each $i\in [k-1]$ and $C_i=(n_1,n_2,j)$ one of the
    following statements is true:
    \begin{enumerate}
    \item $I_j=\inc(1)$ and $C_{i+1}=(n_1+1,n_2,j+1)$.

      This holds iff we can add edges to $w$ so that all
      $\IONEL$-coloured vertices in $\enc(C_i)$ are matched with
      all but one of the $\IONER$-coloured vertices in
      $\enc(C_{i+1})$, and the $\ITWOL$-coloured vertices in
      $\enc(C_i)$ are matched with the $\ITWOR$-coloured
      vertices in $\enc(C_{i+1})$, and the unique instruction vertices
      in $\enc(C_i)$ and $\enc(C_{i+1})$ have the same colour.
      
    \item $I_j=\dec(1,j_0,j_1)$ and either $n_1=0$ and
      $C_{i+1}=(n_1,n_2,j_0)$, or $n_1\geq 1$ and
      $C_{i+1}=(n_1-1,n_2,j_1)$.

      Equivalently, either one of the following statements is true:
      \begin{itemize}
      \item There exists no $\IONEL$-coloured vertex in
        $\enc(C_{i})$ and no $\IONEL$-col\-oured vertex in
        $\enc(C_{i+1})$. Furthermore, the $\ITWOL$-coloured vertices in
        $\enc(C_i)$ can be matched with the $\ITWOR$-coloured
        vertices in $\enc(C_{i+1})$. The unique instruction
        vertex in $\enc(C_{i+1})$ is $j_0$-coloured.
      \item There is at least one $\IONEL$-coloured vertex in
        $\enc(C_{i})$. Furthermore, the $\IONER$-coloured vertices in
        $\enc(C_{i+1})$ can be matched with all but one of the
        $\IONEL$-coloured vertices in $\enc(C_i)$, and the
        $\ITWOL$-coloured vertices in $\enc(C_{i})$ can be
        matched with the $\ITWOR$-coloured vertices in
        $\enc(C_{i+1})$.
        The unique instruction vertex in $\enc(C_{i+1})$ is $j_1$-coloured.
      \end{itemize}
    \item[(c),(d)] Analogous statements to (a), (b) for the case where
      $I_j$ operates on counter $2$.
    \end{enumerate}
  \end{enumerate}
  Now, a \emph{matching extension} of $w_P$ is an ordered graph
  obtained from $w_P$ by adding, for each pair of subsequent
  configurations, exactly the edges of a matching witnessing that
  $w_{P}$ satisfies the conditions (a), (b), (c), and (d). Observe
  that each vertex of a matching extension is contained in at most one
  matching. Hence, any matching extension has maximum degree $1$.
  Using our description above, it is easy to write down a first-order
  sentence $\phi$ defining the class of all matching extensions of
  $w_P$. This class is non-empty iff $P$ halts. Hence $\phi$ is
  $2$-satisfiable iff $P$ halts.
\end{proof}

\subsection{From order-invariant $\textup{FO}[\sigma^{\leq}]$-formulae to $\textup{FO}[\sigma]$-formulae}
\label{sec:oifo-to-fo}

We prove Theorem~\ref{thm:oifo-eq-fo}. The key insight here is that for every
quantifier rank $q$ and every structure $\AS \in \fin_{\sigma,d}$ there exists a
class of canonical linear orders $\preceq_q$ for which the $\FO_q$-type of
$(\AS,\preceq_q)$ is already $\FO$-definable in $\AS$. In particular,
$\tp_q(\AS,\preceq_q)$ only depends on $\AS$, even though there may be more than
one such order on $\AS$.

We call these canonical orders \emph{$q$-orders}. After defining them
formally we will thus prove the following two facts about them:
\begin{enumerate}
\item Expansions by $q$-orders are indistinguishable in
  $\FO_q$, i.e. $(\AS,\preceq_1) \equiv_q (\AS,\preceq_2)$
  for all finite structures $\AS$, provided both $\preceq_1$ and
  $\preceq_2$ are $q$-orders
  (cf.~Lemma~\ref{lem:all-q-orders-are-equivalent}).
\item If the tree-depth of structures is bounded, then the $q$-type
  $\tp_{q}(\AS,\preceq_q)$ of an expansion of $\AS$ by a $q$-order is
  definable in $\FO$ (Lemmas~\ref{lem:connected-lift}
  and~\ref{lem:counting-formulae}). The proof of
  Theorem~\ref{thm:oifo-eq-fo} easily follows from this.
\end{enumerate}

\paragraph{The definition of $q$-orders}

With an eye towards Section~\ref{sec:oimso}, the notion of $q$-orders
will be defined more generally for logics $\logl\in
\set{\FO,\MSO}$. We fix arbitrary orders $\preceq_{\logl,q}$ on the
set of $(\logl,q)$-types over the signature $\sigma^{\leqsym}$, and
$\preceq_{\text{atomic}}$ on the set of atomic $\sigma$-types. For
simplicity we write $a \preceq_{\text{atomic}} b$ for $\alpha(a)
\preceq_{\text{atomic}} \alpha(b)$.

To obtain a $q$-order $\preceq$ on a connected structure $\AS \in
\fin_{\sigma,d}$, we pick a root $r$ of $\AS$ which has
$\preceq_{\text{atomic}}$-minimal atomic type among all roots and for
which the type of $q$-ordered expansions of $\AS^{[r]}$ is
$\preceq_{\logl,q}$-minimal among all
$\preceq_{\text{atomic}}$-minimal roots. We place this $r$ in front of
the order $\preceq$ and order the remaining elements according to a
(recursively obtained) $q$-order on $\AS^{[r]}$. On structures with
more than one component, we $q$-order the components individually and
take the sum of their orders, following the $\preceq_{\logl,q}$-order of
the components:

\begin{definition}[$(\logl,q)$-order]
  \label{def:q-order}
  An \emph{$(\logl,q)$-order} on a $\sigma$-structure $\AS$ is an
  order $\preceq$ which satisfies the following conditions:
  \begin{conditions}
  \item\label{def:q-order-connected} If $\AS$ is connected we denote
    by $r \in A$ its $\preceq$-minimal element. Then either
    $|A|=1$, or $|A|>1$ and the following holds:
    \begin{conditions}
    \item \label{def:q-orders-min-element-is-root} $r$ is a
      $\preceq_{\text{atomic}}$-minimal root of $\AS$, i.e. $r \in
      \tdroot(\AS)$ and $r \preceq_{\text{atomic}} r'$ for all $r' \in
      \tdroot(\AS)$.
    \item \label{def:q-orders-root-q-type} The $(L,q)$-type of
      $q$-ordered expansions of $\AS^{[r]}$ is minimal:
      \[
      \tp_q(\AS^{[r]},\preceq) \preceq_{\logl,q} \tp_q(\AS^{[r']},\preceq')
      \]
      for every $r' \in \tdroot(\AS)$ with $\alpha(r') = \alpha(r)$
      and every $q$-order $\preceq'$ on $\AS^{[r']}$.
    \item\label{def:q-order-connected-restriction}
      $\relao{\preceq}{A\setminus r}$ is an $(\logl,q)$-order on
      $\AS^{[r]}$.
    \end{conditions}
  \item\label{def:q-order-disconnected} If $\AS$ is not connected, we
    denote its components by $\AS_1,\ldots,\AS_\ell$ and set
    $\preceq_i := \relao{\preceq}{A_i}$. Then $\preceq$ is a $q$-order
    if
    \begin{conditions}
    \item each $\preceq_i$ is a $q$-order of $\AS_{i}$, and
    \item after suitably permuting the components,
      \[
      {\preceq} = {\preceq_1} + \cdots + {\preceq_{\ell}}
      \quad\text{and}\quad
      \tp_q(\AS_i,\preceq_i) \preceq_{\logl,q}
      \tp_q(\AS_j,\preceq_j)
      \text{ for }i \leq j.
      \]
    \end{conditions}
  \end{conditions}
  The $\preceq$-minimal element of a $q$-order $\preceq$ will be
  denoted by $r_{\preceq}$.
\end{definition}

It is plain from the definition above that each structure can be
$q$-ordered. Next we want to show that all $q$-ordered expansions
$(\AS,\preceq)$ of a given structure $\AS$ have the same $q$-type, and
that the $q$-type of $(\AS^{[r_\preceq]},\preceq)$ is also the same
for all $q$-orders $\preceq$ of $\AS$.

\begin{lemma}
  \label{lem:all-q-orders-are-equivalent}
  Let $\logl\in \set{\FO{}, \MSO{}}$, $q\in \Npos$. For all
  $(\logl,q)$-orders $\preceq, \preceq'$ of a structure $\AS$, we have
  \[
  (\AS,\preceq) \leleq[q] (\AS,\preceq').
  \]
  If $\AS$ is connected and
  $\td(A) > 1$, then also $(\AS^{[r_\preceq]},\preceq) \leleq[q]
  (\AS^{[r_{\preceq'}]},\preceq')$.
\end{lemma}

For the proof, we will need the following composition lemma for
ordered sums, cf.~\cite{Makowsky2004} for a proof.

\begin{lemma}[Composition Lemma]
  \label{lem:ordered-comp-lemma}
  Let $\logl\in\set{\FO, \MSO}$, $q\in \N$ and let $\sigma$ be a
  relational signature.  Let $(\AS_1,
  \preceq^{\AS_1})$,$(\AS_2,\preceq^{\AS_2})$,$(\BS_1,\preceq^{\BS_1})$,$(\BS_2,\preceq^{\BS_2})$
  be ordered $\sigma$-structures. If 
  \[
  (\AS_1,\preceq^{\AS_1})\leleq[q]
  (\AS_2,\preceq^{\AS_2})
  \quad\text{and}\quad
  (\BS_1,\preceq^{\BS_1})\leleq[q]
  (\BS_2,\preceq^{\BS_2}),
  \]
  then
  \[
  (\AS_1 \disunion \BS_1, \preceq^{\AS_1} + \preceq^{\BS_1})
  \leleq[q] (\AS_2 \disunion \BS_2, \preceq^{\AS_2} +
  \preceq^{\BS_2}).  \]
\end{lemma}

\begin{proof}[Proof of Lemma~\ref{lem:all-q-orders-are-equivalent}]
  The proof proceeds on the size of $A$. If $\absval{A} = 1$ then
  ${\preceq} = {\preceq'}$ and there is nothing to prove.

  Let $\absval{A} > 1$ and suppose first that $\AS$ is connected. By
  Definition~\ref{def:q-order},
  $\alpha(r_{\preceq}) = \alpha(r_{\preceq'})$ and
  \[
  \tp_q(\AS^{[r_{\preceq}]},\preceq) \preceq_{\logl,q}
  \tp_q(\AS^{[r_{\preceq'}]},\preceq').
  \]
  By symmetry also
  \[
  \tp_q(\AS^{[r_{\preceq'}]},\preceq') \preceq_{\logl,q}
  \tp_q(\AS^{[r_{\preceq}]},\preceq),
  \]
  so
  $\tp_q(\AS^{[r_{\preceq}]},\preceq)=\tp_q(\AS^{[r_{\preceq'}]},\preceq')$
  and, by Lemma~\ref{lem:IndASrtoAS}, $(\AS,\preceq) \equiv_q
  (\AS,\preceq')$.

  Now consider the case where $\AS$ is not connected, and let
  $\KS_1,\ldots,\KS_{\ell}$ be the components of $\AS$.
  By the definition of $q$-orders each $\KS_i$ is $q$-ordered, so
  \[
  (\KS_i,\relao{\preceq}{K_i}) \leleq[q]
  (\KS_i,\relao{\preceq'}{K_i})
  \]
  for $i = 1,\ldots,\ell$ by what we have just said. Considering the
  way that an $(\logl,q)$-order orders the components of a structure
  according to their $(\logl,q)$-types
  (Part~\ref{def:q-order-disconnected} of
  Definition~\ref{def:q-order}), we obtain that $(\AS,\preceq)
  \leleq[q] (\AS,\preceq')$ by repeatedly applying the
  Composition Lemma.
\end{proof}

By Lemma~\ref{lem:all-q-orders-are-equivalent} it makes sense to speak
of the \emph{$q$-order type} of an unordered structure $\AS$ which we
define as $\tp^\leqsym_q(\AS) := \tp_q(\AS,\preceq_q)$
If $\AS$ is connected and $\td(\AS) > 1$, we furthermore define its
\emph{$q$-order root type} as $\rtp^\leqsym_q(\AS) :=
\tp_q(\AS^{[r_{\preceq_q}]},\preceq_q)$. 
In both cases $\preceq_q$ is some $q$-order on $\AS$ and
well-definedness is guaranteed by the Lemma. Note that both these
types are $\sigma^{\leqsym}$-types. Similarly, the atomic type $\alpha_\AS :=
\alpha(r_\leq)$ of the minimal element in a $q$-ordered expansion of $\AS$ is 
well-defined.

We set
\[\begin{split}
\types[\logl,\sigma,q,d] &:= \{ \tp^\leqsym_q(\AS) \st \AS \in
\fin_{\sigma,d} \},
\\
\ctypes[\logl,\sigma,q,d] &:= \{ \tp^\leqsym_q(\AS) \st \AS \in
\fin^\conn_{\sigma,d} \},\text{ and}
\\
\types[\logl,\sigma,q] &:= \bigcup_{d \in \Npos}
\types[\logl,q,\sigma,d].
\end{split}
\]
We say that a sentence $\phi_{\tau} \in \logl[\sigma]$ \emph{defines
  $\tau$ on $\fin_{\sigma,d}$} (and that $\tau$ is $\logl$-definable)
if for each $\AS\in\fin_{\sigma,d}$, we have
\[
\AS\models \phi_{\tau}
\quad\text{iff}\quad
\tp^\leqsym_{q}(\AS)=\tau.
\]
Note that the sentence $\phi_\tau$ must not contain the relation
$\leqsym$.

By Lemma~\ref{lem:IndASrtoAS} the atomic type of $r_\preceq$ and the
$q$-type of $\AS^{[r_\preceq]}$ determine the $q$-type of $\AS$, and
$\td(\AS^{[r_{\preceq}]}) = \td(\AS) - 1$, 
for connected structures $\AS$ and $q$-orders $\preceq$. Since the
number of atomic $\tilde \sigma$-types is $2^{\vert{\tilde
    \sigma}\vert}$, we obtain the following bound on the size of
$\ctypes[\sigma,q,d]$:

\begin{corollary}
  \label{cor:num-connected-types}
  Let $q,d \in \Npos$. Then $\vert\ctypes[\sigma, q,d]\vert \leq
  2^{\vert{\tilde \sigma}\vert} \cdot \vert{\types[\tilde \sigma,
    q,d-1]}\vert.$
\end{corollary}

\subsection{Handling connected structures}

The proof of our main theorem is broken down into two steps. In the first step,
we show how to lift the definability of $q$-types of $q$-ordered structures from
structures of tree-depth $d-1$ to connected structures of tree-depth $d$.

Again we invoke Lemma~\ref{lem:IndASrtoAS} and
Lemma~\ref{lem:all-q-orders-are-equivalent} to show that $q$-order types can be
broken down into atomic types of roots and $q$-order root types:

\begin{corollary}
  \label{cor:rtp-to-tp}
  Let $d > 1$ and let $\tau\in\ctypes[\sigma, q,d]$.  Let
  \[
  R_{\tau}\ :=\ \{(\alpha_\AS, \rtp_q^\leqsym(\AS)) \st
  \AS \in \fin^{\conn}_{\sigma,d}, \td(\AS) > 1,\text{ and
  }\tp_q^\leqsym(\AS)=\tau \}.
  \]
  Then for each $\BS\in \fin^{\conn}_{\sigma,d}$, we have
  $\tp_q^\leqsym(\BS)=\tau$
  iff
  $(\alpha_\BS, \rtp_q^\leqsym(\BS)) \in  R_{\tau}$.
\end{corollary}
\begin{proof}
  The ``only-if''-part of the claim is obvious.  Regarding the
  ``if''-part, if 
  \[
  (\alpha_\BS, \rtp_q^\leqsym(\BS))= 
  (\alpha_\AS, \rtp_q^\leqsym(\AS))
  \]
  for some $\AS$ with $\tp^\leqsym_q(\AS)=\tau$,
  then Lemma~\ref{lem:all-q-orders-are-equivalent} and the definitions
  of $\tp^\leqsym_q,\rtp^\leqsym_q$ imply that
  $\tp_q^\leqsym(\BS)=\tau$.
\end{proof}

\begin{lemma}
  \label{lem:connected-lift}
  Let $q,d \in \Npos$ with $d > 1$. Let $(\logl_1,\logl_2)$ be one of
  $(\FO,\FO)$ or $(\MSO,\FOmod)$. If each $(\logl_1,q)$-type
  $\theta\in\types[\tilde \sigma, q,d-1]$ is
  $\logl_2[\tilde\sigma]$-definable on $\fin_{\tilde\sigma,d-1}$ by a sentence
  $\psi_{\theta,d-1}$, then each $(\logl_1,q)$-type
  $\tau\in\ctypes[\sigma,q,d]$ is $\logl_2[\sigma]$-definable on
  $\fin^{\conn}_{\sigma,d}$ by a sentence $\phi^{\conn}_{\tau,d} $.
  Moreover, defining
  \[
  \Psi:=\setc{\psi_{\theta,d-1}}{\theta\in\types[\tilde \sigma,
    q,d-1]}
  \quad\text{and}\quad
  \Phi:=\setc{\phi^\conn_{\tau,d}}{\tau\in\types[\sigma, q,d]},
  \]
  we have $\size{\Phi} \leq c \cdot \size{\Psi} \cdot
  \absval{\types[\tilde \sigma, q,d-1]}^2$ and $\qad{\Psi} \leq
  \qad{\Phi} + 1$, for a constant $c$ depending only on $\sigma,d$.
\end{lemma}
\begin{proof}
  In the following, all $q$-types are
  $(\logl_1,(\sigma^{\leq}),q)$-types.  Let
  $\tau\in\ctypes[\sigma,q,d]$ and let $R_\tau$ be as in
  Corollary~\ref{cor:rtp-to-tp}. We show that, under the assumptions
  of our lemma, the class
  \[ \setc{\AS\in\fin^{\conn}_{\sigma,d}}{(\alpha_\AS,
    \rtp_q^\leqsym(\AS)) \in R_{\tau}}\] is
  $\logl_2[\sigma]$-definable by a sentence $\phi_\tau$ on
  $\fin^{\conn}_{\sigma,d}$. Taking care of connected structures of
  tree-depth $1$ (i.e. singleton structures) we set
  $
  \phi^{\conn}_{\tau,d}\ :=\ (\td_{\leq 1}\land \hat\phi_{\tau}) \lor
  (\td_{> 1} \land \phi_\tau),
  $
  where $\hat\phi_\tau$ defines $\tau$ on singleton structures.

  For each atomic $\sigma$-type $\alpha\subseteq\sigma$, the following
  $\FO$-sentence $\xi_\alpha$ expresses in a structure
  $\AS\in\fin^{\conn}_{\sigma,d}$ that $\alpha_\AS=\alpha$:
  \[ \xi_\alpha\ := \left(\exists x\,\big(\tdroot_d(x) \wedge \alpha(x)\big)
    \right) \wedge \left(\forall x\, \big(\tdroot_d(x)
  \limplies \smashoperator{\biglor_{\alpha \atleq \alpha'}}
  \alpha'(x)\big)\right).
  \]
  For each type $\theta\in\types[\tilde \sigma, q,d-1]$ the following sentence
  is true in a $\sigma$-structure $\AS$ if, and only if, there is a root $r$ of
  atomic type $\alpha$ for which $\AS^{[r]}$ has type $\theta$, and $\theta$ is
  $\preceq_{\logl_1,q}$-minimal among the types of $\AS^{[s]}$ for roots $s$ of
  atomic type $\alpha$:
  \begin{align*}
    \chi_{\alpha,\theta}\ :=\ &\forall x\,\Big(
  (\tdroot_d(x) \land \alpha(x)) \limplies
  \smashoperator{\biglor_{\theta \preceq_{\logl_1,q} \theta'}}
  \I(\psi_{\theta',d-1})(x)\Big)\\
   \land\  &\exists x\ \big(\tdroot_d(x) \land \alpha(x) \land \I(\psi_{\theta,d-1})(x)\big).
  \end{align*}
  Observe that $\qad{\chi_{\alpha,\theta}} \leq \qad{\Psi} + 1$.
  
  Now we obtain the desired sentence by
  defining $\phi_{\tau}\ :=\
  \smashoperator{\biglor_{(\alpha,\theta)\in R_\tau}}
  \big(\xi_{\alpha} \land \ \chi_{\alpha,\theta})$.

  Observe that, for some constant $c$ depending only on $\sigma$, $d$, we have
  $\size{\xi_\alpha} \leq c$, $\size{\chi_{\alpha,\theta}} \leq c \cdot
  \size{\Psi} \cdot \absval{\types[\tilde \sigma, q,d-1]}$, $\absval{R_\tau}
  \leq c \cdot \absval{\types[\tilde \sigma, q,d-1]}$, and  
  $\size{\phi_\tau} \leq c \cdot \size{\Psi} \cdot \absval{\types[\tilde
    \sigma, q,d-1]}^2$.
  The claims about $\size{\Phi}$ and $\qad{\Phi}$ follow from the observations above. 
\end{proof}

\subsection{Handling disconnected structures}

We proceed with the preparations for the second step in the proof of
our main theorem, where we lift the definability of $q$-order types
from connected structures of tree-depth $\leq d$ to disconnected
structures of tree-depth $\leq d$.

For us, a \emph{Boolean query} is an isomorphism-invariant map $f : \fin \to
\set{0,1}$, where $\fin$ is the class of all finite structures (i.e. structures
over arbitrary signatures). We will treat maps $f : \fin_\sigma \to \set{0,1}$
as Boolean queries by assuming that $f(\AS)=0$ if $\AS$ is not a
$\sigma$-structure. The general definition for arbitrary signatures will be
useful in in Section~\ref{sec:mso} below.  We are interested in two kinds of
queries.  As usual, we identify each sentence $\phi$ with a Boolean query such
that $\phi(\AS)=1$ iff $\AS\models\phi$.  Furthermore, we identify each
$q$-order type $\tau$ with a query such that $\tau(\AS)=1$ iff
$\tp^\leqsym_q(\AS)=\tau$. For each structure $\AS$ and each Boolean query $f$,
we let $n_f(\AS)$ denote the number of components $\KS$ of $\AS$ such that
$f(\KS)=1$.  For each ordered set $Q := \set{f_1, \ldots, f_\ell}$ of Boolean
queries, we let $\bar n_Q(\AS):=(n_{f_1}(\AS), \ldots, n_{f_\ell}(\AS))$.  For
natural numbers $a, b, t \in \Npos$ we set
\[
a \maxeq{t} b\quad
\Leftrightarrow
\quad
(a = b \text{ or }a,b \geq t),
\]
and we extend this relation to tuples $\bar a$ and $\bar b$ by saying
$\bar a \maxeq{t} \bar b$ if, and only if, $a_i \maxeq{t} b_i$ for
all components $a_i$ and $b_i$.

We show that \FO{} inherits its capability to count the types of
components in $q$-ordered structures from its capability to
distinguish linear orders of different length.  The proof of the
following lemma closely follows a step in the proof of
\cite[Thm. 5.5]{BenediktSegoufin2009}.
Observe that for all $\AS, \BS \in \fin_{\sigma,d}$, 
$n_{\ctypes[\sigma,q,d]}(\AS) \maxeq{t} n_{\ctypes[\sigma,q,d]}(\BS)$ iff 
$n_{\types[\sigma,q]}(\AS) \maxeq{t} n_{\types[\sigma,q]}(\BS)$. 

\begin{lemma}
  \label{lem:cut-determines-type}
  Let $d \geq 1$, $q\in \Npos$ and $t:=2^q+1$. Then for all $\AS,\BS\in
  \fin_{\sigma,d}$, 
  \[ n_{\types[\sigma,q]}(\AS) \maxeq{t} n_{\types[\sigma,q]}(\BS) \ \Longrightarrow \ 
  \tp^\leqsym_q(\AS)=\tp^\leqsym_q(\BS).
  \]
\end{lemma}
\begin{proof}
  For each component $\KS$ of $\AS$, we let $\preceq^\KS$ be a
  $q$-order of $\KS$.
  By Part~\ref{def:q-order-disconnected} of Definition~\ref{def:q-order},
  the $q$-orders on the components of $\AS$ can be extended to a $q$-order
  $\preceq^\AS$ on $\AS$ such that
  $\relao{\preceq^\AS}{\KS}=\preceq^\KS$ for each component $\KS$ of
  $\AS$. We proceed analogously to obtain a $q$-order $\preceq^\BS$ on
  $\BS$.  Let $\types[\sigma,q]=\set{\tau_{1}, \ldots, \tau_{\ell}}$,
  where $\ell:=|\types[\sigma,q]|$ and $\tau_{i} \preceq_{q} \tau_{j}$
  iff $i<j$.
  We consider words over the alphabet $\types[\sigma,q]$ as structures in the usual way,
  i.e. as ordered structures over a signature containing a unary relation symbol for each type.
  Consider the words $w_{\AS},w_{\BS}\in\types[\sigma,q]^{*}$ obtained from $(\AS,\preceq^\AS)$ and
  $(\BS,\preceq^\BS)$ by contracting each component $\KS$ to a single
  element that gets labelled by its $q$-type in the corresponding
  $q$-ordered structure.
  By this construction and by Part \ref{def:q-order-disconnected} of
  Definition~\ref{def:q-order}, we know that 
  \[
  w_{\AS}=\tau_{1}^{n_{\tau_{1}}(\AS)} \dotsb
  \tau_{\ell}^{n_{\tau_{\ell}}(\AS)} \quad \text{and} \quad
  w_{\BS}=\tau_{1}^{n_{\tau_{1}}(\BS)} \dotsb
  \tau_{\ell}^{n_{\tau_{\ell}}(\BS)}.
  \]
  Since $n_{\types[\sigma,q]}(\AS) \maxeq{t}
  n_{\types[\sigma,q]}(\BS)$, for each $i\in [\ell]$, we have either
  $n_{\tau_{i}}(\AS) = n_{\tau_{i}}(\BS)$ or
  $n_{\tau_{i}}(\AS),n_{\tau_{i}}(\BS) \geq t$. A folklore result
  (cf. \cite[Ch. 3]{Libkin2004}) tells us that $w_{\AS} \foeleq[q]
  w_{\BS}$, i.e. Duplicator has a winning strategy in the $q$-round
  EF-game on the two word structures.

  We show that $(\AS,\preceq^\AS)\foeleq[q](\BS,\preceq^\BS)$.  To
  this end, consider the following winning strategy for Duplicator in
  the $q$-round EF-game on $(\AS,\preceq^\AS)$ and
  $(\BS,\preceq^\BS)$. She maintains a \emph{virtual} $q$-round
  EF-game $w_{\AS}$ on $w_{\BS}$ between a \emph{Virtual Spoiler} and
  a \emph{Virtual Duplicator}. When, during the $i$-th round, Spoiler
  chooses an element $v$ in some component $\KS$ of, say, $\AS$, she
  lets the Virtual Spoiler play the corresponding position in
  $w_{\AS}$ in the $i$-th round of the virtual game. The Virtual
  Duplicator answers in $w_{\BS}$. Duplicator chooses a component
  $\KS'$ of $\BS$ for its reply according to the Virtual Duplicator's
  answer in $w_{\BS}$. The winning strategy on $w_{\AS}$ and $w_{\BS}$
  ensures that $(\KS,\preceq^\AS) \foeleq[q] (\KS',\preceq^\BS)$ and
  that all elements of $\KS$ and $\KS'$ have the same positions in
  $\preceq^\AS$ and $\preceq^\BS$ relative to the elements played in
  the previous rounds. Duplicator uses her winning strategy in the
  $q$-round game on the ordered components to determine the element of
  $\KS'$ that she uses as her answer to $v$.
\end{proof}

For a tuple $\bar a$ of natural numbers, denote by $\redmax{\bar a}{t}$ the
tuple obtained from it by replacing all entries $> t$ with $t$.  Then
the previous lemma implies that if $\td(\AS) \leq d$, then
$\redmax{\bar n_{\types^\conn}(\AS)}{(2^q + 1)}$ determines
$\tp^{\leqsym}(\AS)$. Hence we obtain the following corollary:

\begin{corollary}
  \label{cor:tau-set}
  Let $q,d\in \Npos$ and let $t:=2^q+1$. For each $\phi \in \FO[\sigma^{\leq}]$, let
  \[
  R_\phi \ := \ \{ \redmax{\bar n_{\types^\conn}(\AS)}{t} \st \AS\in
  \fin_{\sigma,d}, \tp^\leqsym_q(\AS)\models \phi\}. 
  \]
  Then for each $\AS\in \fin_{\sigma,d}$, we have
  \[
  \tp^\leqsym_q(\AS)\models \phi
  \quad\text{if, and only if, }
  \redmax{\bar n_{\types^\conn}}{t} \in
  R_\tau.
  \]
  Furthermore, $\absval{\types[\sigma,q,d]} \leq
  (t+1)^{\absval{\types^\conn}}$.
\end{corollary}

The following lemma will be used in conjunction with the previous corollary to
lift the definability of $q$-types from connected to disconnected structures.
\begin{lemma}
  \label{lem:counting-formulae}
  Let $\logl\in \set{\FO{}, \FOmod{}}$. For all $d,t\in \Npos$, every
  set of $\logl$-sen\-tences $\Phi$, and every set
  $R\subseteq[0,t]^{\vert \Phi \vert}$, there is an $\logl$-sentence
  $\psi^{\Phi}_{R}$ such that for each structure $\AS$ with $\td(\AS)
  \leq d$, we have
  \[ \AS\models \psi^{\Phi}_{R} \iff \redmax{\bar n_{\Phi}(\AS)}{t} \in R. \]
  Moreover, 
    $\size{\psi^{\Phi}_{R}} \leq c \cdot
  \absval{\Phi} \cdot \size{\Phi}\cdot \vert R \vert \cdot t^2$ and
  $\qad{\psi^{\Phi}_{R}} \leq \qad{\Phi} + 2$,
for a constant $c$ which depends only on $\sigma,d$.
\end{lemma}
\begin{proof}
  Let $\Phi := \set{\phi_1, \ldots, \phi_\ell}$.
  Consider some $i\in
  [\ell]$ and let ${\tilde\phi}_{i}(x):=\rela{\phi_{i}}{\reach_{d}(x,z)}$.

  Let $n\in [t]$. We define a formula $\psi^{n}_{i}(\bar{x})$, where
  $\bar{x}:=(x_1, \ldots, x_n)$, which states that $x_1,\ldots,x_n$
  lie in distinct connected components, each of which satisfies
  $\phi_i$:
  \begin{align*}
    \psi^{n}_{i}(\bar{x}) \ := \ \smashoperator{\bigland_{j \in
        [n]}} {\tilde\phi}_{i}(x_j) \ \land \
    \smashoperator{\bigland_{j,k \in [n],\, j\neq k}} \lnot
    \reach_d(x_j,x_k).
  \end{align*}
  Observe that $\qad{\psi^{n}_{i}} \leq \qad{\Phi}$ (in particular,
  since $\reach_d$ is an existential formula) and that
  $\size{\psi^{n}_{i}} \leq c n^2 \size{\Phi} \leq c t^2 \size{\Phi}$,
  for a constant $c$ depending on $\sigma,d$ only.

  To obtain a formula which states that either the (pairwise disjoint)
  components of the $x_1, \ldots, x_n$ are the only components which
  satisfy $\phi_i$ or the number of such components is at least $t$,
  we let
  \[
  \psi^{n,t}_{i}(\bar{x})\ :=\
  \begin{cases}
    \forall y\ \lnot{\tilde\phi}_{i}(y) &\text{if }n = 0,
    \\
    \psi^{n}_{i}(\bar{x}) \land \forall y\,({\tilde\phi}_{i}(y)
    \limplies \biglor_{i\in [n]} \reach_d(y,x_i)) &\text{if $0 <
      n < t$}
    \\
    \psi^{n}_{i}(\bar{x}) &\text{if $n \geq t$}.
  \end{cases}
  \]
  Note that $\qad{\psi^{n,t}_{i}} \leq \qad{\Phi} + 1$ and
  $\size{\psi^{n,t}_{i}} \leq c \cdot \size{\psi^{n}_{i}}$, for some constant $c$
  depending on $\sigma,d$ only. (Note that $\size{\psi^n_i} \geq n$, so the
  disjunction over $i \in [n]$ is absorbed by that.) We obtain the desired
  sentence $\psi^{\Phi}_{R,t}$ by setting 
  \[
  \psi^{\Phi}_{R,t} := \biglor_{(n_1, \ldots, n_\ell)\in R} \
  \exists {\bar x}_i\,\bigland_{i \in [\ell]}  \psi^{n_i,t}_{i}({\bar x}_i),
  \]
  where ${\bar x}_i$ is a tuple of $n_i$ variables. Note that
  \[\begin{split}
    \size{\psi^{\Phi}_{R}} & \ \leq \ \absval{R} \cdot \absval{\Phi} \cdot \max_{i\in [\ell]} \size{\psi^{t}_{i}} \ \leq \ c
    \cdot \absval{R} \cdot \absval{\Phi} \cdot \size{\Phi} \cdot  t^2,
    \\
    \qad{\psi^{\Phi}_R} &\ \leq \ \max_{i\in [\ell]} \qad{\psi^{n_i,t}_{i}} + 1
    \ \leq \ \qad{\Phi} + 2\, . 
  \end{split}
  \]%
\end{proof}

Finally, we can prove our main theorem.

\begin{proof}[Proof of Theorem~\ref{thm:oifo-eq-fo}]
  By induction on the tree-depth $d$, 
  we show that for each signature $\sigma$ and each $\FO[\sigma^{\leq}]$-sentence $\phi$ with $\qr{\phi} = q$, there is an $\FO[\sigma]$-sentence $\psi_{\phi,d}$ with $\size{\psi_{\phi,d}}\in \nexp[d](q)$ and $\qad{\psi_{\phi,d}}\leq 3d$ such that for each $\AS\in \fin_{\sigma,d}$,
  we have $\AS\models \psi_{\phi,d}$ iff $\tp^\leqsym_q(\AS)\models \phi$.
  Furthermore, we show that $\absval{\types[\sigma,q,d]}\in \nexp[d](q)$ and $\absval{\ctypes[\sigma,q,d]}\in \nexp[(d-1)](q)$.
  To finish the proof, if $\phi$ is order-invariant, we let $\psi := \psi_{\phi,d}$, and we obtain that $\AS \models_{\leqsym} \phi$ iff $\AS \models \psi$. 

  Let $\ctypes[\sigma, q, d]=\set{\theta_{1}, \ldots, \theta_{\ell}}$.
  First, for each $i\in [\ell]$, we construct a sentence $\phi_i$ that
  defines $\theta_i$ on $\fin^\conn_{\sigma,d}$. If $d=1$, observe
  that any \emph{connected} structure $\AS$ of
  type~$\theta_i\in\ctypes[\sigma,q,1]$ consists of a single element. The atomic $\sigma$-type $\alpha$
  of this element determines the $q$-type of the unique $q$-order on $\AS$.
  The $\FO{}$-sentence $\phi^{\conn}_{\tau,1} \ := \ \exists x\, \alpha(x)$
  hence defines $\tau$ on $\fin^\conn_{\sigma,1}$. We obviously have
  $\size{\phi^{\conn}_{\tau,1}} \leq c \cdot \absval{\sigma}$, for some absolute
  constant $c$, and $\absval{\ctypes} \leq 2^{\absval{\sigma}} \in \nexp[(d-1)](q)$ .

  If $d>1$, we construct an $\FO$-sentence $\psi_{\theta,d-1}$ inductively for each $q$-type $\theta\in\types[\tilde \sigma, q, d-1]$.
  Let $\Psi := \setc{\psi_{\theta,d-1}}{\theta\in\types[\tilde \sigma, q, d-1]}$. 
  By induction, we obtain $\size{\Psi} \in \nexp[(d-1)](q)$, and $\qad{\Psi} \leq 3(d-1)$, and we have $\absval{\types[\tilde \sigma, q, d-1]} \in \nexp[(d-1)](q)$.
  We construct $\phi_i$ according to Lemma~\ref{lem:connected-lift},
  i.e. we let $\phi_i:=\phi^{\conn}_{\theta_i,d}$ for each $i\leq
  \ell$. Let $\Phi := \set{\phi_1, \ldots, \phi_\ell}$.
  Then there is a constant $c$ depending only on $\sigma,d$,
  such that 
  \[
  \begin{split}
    \size{\Phi} &\leq c \cdot \size{\Psi} \cdot \absval{\types[\tilde \sigma, q, d-1]}^2 \ \in \ \nexp[(d-1)](q)
    \quad\text{and}
    \\
    \qad{\Phi} &\leq \qad{\Psi} + 2 \ \leq \ 3(d-1) +2.
  \end{split}
  \]

  Now consider a sentence $\phi\in\FO[\sigma^{\leq}]$.  Let
  $R:=R_\phi$ be given by Corollary~\ref{cor:tau-set}. We apply
  Lemma~\ref{lem:counting-formulae} with $t:=2^{q}+1$ to obtain a
  sentence $\psi_{\phi,d}:=\psi^{\Phi}_{R}$.  To see that
  $\psi_{\phi,d}$ is defined correctly, consider some $\AS\in
  \fin_{\sigma,d}$.  Observe that for each $i\in [\ell]$ and each
  component $\KS$ of $\AS$, we have $\KS \models \phi_{i}$ iff
  $\tp^\leqsym_q(\KS)=\tau_i$, and thus $\bar n_{\Phi}(\AS)=\bar
  n_{\ctypes[\sigma,q,d]}(\AS)$.  Then
  \begin{align*}
  \AS \models \psi_{\phi,d}
  &
  \text{ iff }
  \redmax{\bar n_{\ctypes[\sigma,q,d]}(\AS)}{t}\in R
  &&
  \text{(by Lemma~\ref{lem:counting-formulae} and previous observation)}
  \\
  &
  \text{ iff }
  \tp^\leqsym_q(\AS)\models \phi.
  &&
  \text{(by Corollary~\ref{cor:tau-set})}
  \end{align*}
  
  By Lemma~\ref{lem:counting-formulae}, for some constant $c$ depending only on $\sigma,d$, we
  have
  \[
  \begin{split}
    \size{\psi_{\phi,d}} &\leq c \cdot \absval{\Phi} \cdot \absval{R} \cdot t^2 \cdot \size{\Phi}
    \quad\text{and}
    \\
    \qad{\psi_{\phi,d}} &\leq \ \qad{\Phi} + 1 \ \leq \ 3d\, .
  \end{split}
  \]
  Observe that $\absval{\Phi} = \ell = \absval{\ctypes} \in
  \nexp[(d-1)](q)$ by Corollary~\ref{cor:num-connected-types} and that
  $\absval{R} \leq t^{\ell} \in \nexp[d](q)$. Hence,
  $\size{\psi_{\phi,d}} \in \nexp[d](q)$.
  By Corollary~\ref{cor:tau-set}, we also obtain
  $\absval{\types}\in \nexp[d](q)$.
\end{proof}

\section{Order-invariant monadic second-order logic}
\label{sec:oimso}

\cite[Thm.~4.1]{Courcelle1996} proved that classes of graphs
definable by order-invariant MSO sentences are recognisable. Recognisable sets
of graphs of \emph{bounded tree-width} are conjectured in~\cite[Conjecture
1]{Courcelle1991} to be definable in \MSO{} with modulo-counting (\CMSO{}),
which would imply that $\oiMSO{}$ is equivalent to $\CMSO$ on these graphs. Note
that it is well-known and easy to see that, regardless of the considered class
of structures, for each sentence of modulo-counting \MSO{} there is an
equivalent \oiMSO{}-sentence. Hence, the difficult part is the construction of
an \CMSO{}-sentence for a given \oiMSO{}-sentence.

While the equivalence of recognisability and definability in \CMSO{} for graphs
of bounded tree-width is still widely considered to be open
(cf.~\cite[p.~574]{CourcelleE2012}), we show that in the further restricted case
of structures of bounded tree-depth, $\oiMSO$ collapses even to
\emph{first-order} logic with modulo counting (\FOmod{}):

\begin{theorem}
  \label{thm:oimso-eq-fomod}
  For every $d\in \Npos$ and every $\oiMSO{}$-sentence $\phi$ there is an
  $\FOmod{}$-sentence $\psi$ with $\qad{\psi} \leq 3d$
  which is equivalent to $\phi$ on $\fin_{\sigma,d}$.
\end{theorem}

In contrast to the previous sections, we do not analyse the formula size,
because it is known from \cite{GroheSchweikardt05} that (plain) \MSO{} can
define the length of orders non-elementarily more succinct than \FO{}.

For the proof of Theorem~\ref{thm:oimso-eq-fomod}, we proceed similarly to the
last section. Again we need to understand $\oiMSO{}$'s capabilities to count the
number of components of a given $q$-type in $q$-ordered structures. However,
this time we need to count not only up to some threshold, but also modulo some
fixed divisor.

For $n\in \N$ and $p\in \Npos$, we let $\redmod{n}{p}$ denote the remainder of
the division of $n$ by $p$, and $\bar n := (n_1, \ldots, n_\ell)\in \N^\ell$, we
let $\redmod{\bar n}{p} := (\redmod{n_1}{p},\ldots, \redmod{n_\ell}{p})$.
Similarly, we set $m \modeq{p} n$ if $p$ divides $m-n$, and extend this notion
to tuples $\bar m$ and $\bar n$ component-wise.

Below, we prove the following Lemma which shows that $\MSO{}$ inherits its
component counting capabilities in $q$-ordered structures from its
capabilities to distinguish orders of different lengths.

\begin{lemma}
  \label{lem:oimso-cut}
  For each $q\in \Npos$, there is a $p\in \Npos$ such that for all
  $q$-ordered structures $(\AS,\preceq^\AS)$ and $(\BS,\preceq^\BS)$,
  \[ 
  \left(
    {\bar n}_{\types[\sigma,q]}(\AS) \modeq{p} {\bar n}_{\types[\sigma,q]}(\BS)
    \text{ and }
    {\bar n}_{\types[\sigma,q]}(\AS) \maxeq{p} {\bar n}_{\types[\sigma,q]}(\BS)
  \right)
  \implies \ (\AS,\preceq^\AS) \msoeleq[q] (\BS,\preceq^\BS) . \]
\end{lemma}

In the following, we say that an ordered structure $(\AS,\preceq)$ is
\emph{component ordered}, if the order $\preceq$ is a sum of the
orders on the components of $\AS$, i.e.  for some enumeration
$\KS_{1}, \ldots, \KS_{n}$ of the components of $\AS$, we have
$\preceq = \relao{\preceq}{K_{1}} + \relao{\preceq}{K_{2}} + \ \dotsb
\ + \relao{\preceq}{K_{n}}$.  Observe that $q$-ordered structures are
also component ordered.  It will be convenient to have some notation
that allows us to treat component ordered structures similarly to
words. Given two ordered structures $(\AS,\preceq^\AS)$ and
$(\BS,\preceq^\BS)$, we let $(\AS,\preceq^\AS) \disunion
(\BS,\preceq^\BS):=(\AS \disunion \BS, \preceq^\AS + \preceq^\BS)$,
where $\AS \disunion \BS$ denotes the disjoint union of $\AS$ and
$\BS$ and we consider $\preceq^\AS, \preceq^\BS$ as orders on the
components of the disjoint union (via the inclusion mappings for $\AS,
\BS$).  Instead of $(\AS,\preceq^\AS) \disunion (\BS,\preceq^\BS)$, we
also write $(\AS,\preceq^\AS)(\BS,\preceq^\BS)$.  Like in the
following definition, we often omit the order to make this notation
less cluttered. For each component ordered structure $\AS$, we define
its \emph{$i$-th power} $\AS^{i}$ by $\AS^{1}:=\AS$ and
$\AS^{i}:=\AS^{i-1}\AS$ if $i >1$.

The proof of Lemma~\ref{lem:oimso-cut} rests on the following Lemma.

\begin{lemma}[Pumping Lemma]
  \label{lem:pumping-lemma}
  For each $q\in \Npos$, there is a number $p\in \Npos$ such that for
  all component ordered structures $\AS$ and all $r\in \N$, $i,j\in \Npos$, 
  \[ \AS^{r+ip} \msoeleq[q] \AS^{r+jp}. \]
\end{lemma}
\begin{proof}
  Let $\types[]$ denote the (finite) set of $q$-types which are
  realised by component ordered $\sigma$-structures.  We lift the
  disjoint union of ordered structures to $\types[]$ by defining
  $\tp_q(\AS) \disunion \tp_q(\BS) := \tp_q(\AS \disunion \BS)$.  The
  Composition Lemma (Lemma~\ref{lem:ordered-comp-lemma}) shows that
  this operation is well-defined.  It is also associative, so that
  $(\types[],\disunion)$ is a finite semigroup.  Hence, there is a
  number $p$ such that for each $\tau\in \types[]$, $\tau^p$ is
  idempotent (cf. e.g. \cite{Howie1976}), i.e. $\tau^p = \tau^{ip}$
  for each $i\in \Npos$.  Then, for all $\AS, r,i,p$ as in the
  statement of the lemma, $\tp_q(\AS)^{r + ip} = \tp_q(\AS)^{r+jp}$,
  i.e.  $\AS^{r+ip} \msoeleq[q] \BS^{r+jp}$.
\end{proof}

\begin{proof}[Proof of Lemma~\ref{lem:oimso-cut}]
  Let $\types[\sigma,q] = \set{\tau_{1}, \ldots, \tau_{\ell}}$ with
  $\tau_{i} \prec_{q} \tau_{j}$ iff $i < j$.  For each $i\in [\ell]$,
  fix a connected $q$-ordered structure $\KS_{i}$ whose type is
  $\tp_{q}(\KS_{i}) = \tau_{i}$.  By repeated application of the
  Composition Lemma, we can assume without loss of generality that
  $\KS\iso \KS_{i}$ for each $q$-ordered component $\KS$ of $\AS$ or
  $\BS$ with $\tp_{q}(\KS)=\tau_{i}$.  Let $n_{i}:=n_{\tau_{i}}(\AS)$
  and let $m_{i}:=n_{\tau_{i}}(\BS)$ for each $i\in [\ell]$. By
  part~\ref{def:q-order-disconnected} of Definition~\ref{def:q-order},
  we obtain
  \[
  \AS \iso \KS_{1}^{n_{1}} \KS_{2}^{n_{2}} \dotsb
  \KS_{\ell}^{n_{\ell}} 
  \quad \text{ and } \quad
  \BS \iso
  \KS_{1}^{m_{1}} \KS_{2}^{m_{2}} \dotsb \KS_{\ell}^{m_{\ell}}.
  \]
  For each $i\in [\ell]$, we have $n_{\tau_i}(\AS) \modeq{p}
  n_{\tau_i}(\BS)$, i.e. there are $r_{i}\in [0,p-1]$ and
  $a_{i}, b_{i}\in \N$ such that $n_{i}=r_{i}+a_{i}p$ and
  $m_{i}=r_{i}+b_{i}p$. Furthermore, as
  $n_{\tau_i}(\AS) \maxeq{p} n_{\tau_i}(\BS)$, we have $a_{i}>0$ iff
  $b_{i}>0$.  By repeated application of the Pumping Lemma, we obtain
  \[ \KS_{1}^{n_{1}} \KS_{2}^{n_{2}} \dotsb \KS_{\ell}^{n_{\ell}} \
  \msoeleq[q] \ \KS_{1}^{r_{1}+b_1p} \KS_{2}^{r_{2}+b_2p} \dotsb \KS_{\ell}^{r_{\ell}+b_\ell p}
  \ = \ \KS_{1}^{m_{1}} \KS_{2}^{m_{2}} \dotsb
  \KS_{\ell}^{m_{\ell}}. \] Hence, $\AS \msoeleq[q] \BS$.
\end{proof}

The next lemma is a modulo-counting analogue of
Lemma~\ref{lem:counting-formulae}.

\begin{lemma}
  \label{lem:mod-counting-formulae}
  For all $d,p\in \Npos$, each set of $\FOmod[\sigma]$-sentences
  $\Phi$, and each
  set $R \subseteq [0,p]^\ell \times [0,p-1]^\ell$, there is an
  $\FOmod[\sigma]$-sentence $\chi^{\Phi}_{R}$ such that for each $\AS\in \fin_{\sigma,d}$,
  \[
  \AS \models \chi^\Phi_R
  \quad\text{iff}\quad
  \left(
    \redmax{\bar{n}_{\Phi}(\AS)}{p},
    \redmod{\bar{n}_{\Phi}(\AS)}{p}
  \right) \in R.
  \]
  Furthermore,
  $
  \qad{\chi^\Phi_R}\leq \max\set{\qad{\Phi}+2, 2(d-1)+1}
  $.
\end{lemma}

In contrast to Lemma~\ref{lem:counting-formulae}, the proof of
Lemma~\ref{lem:mod-counting-formulae} is not straightforward, because it is not
obvious how modulo-counting quantifiers can be used to count the number of
components satisfying a given $\FOmod{}$-sentence.  A remedy to this problem is
provided by the following Lemma~\ref{lem:roots}, which shows that the number of
tree-depth roots of each component of a graph (and hence of a structure) can be
bounded in terms of its tree-depth only.

\begin{proof}[Proof of Lemma~\ref{lem:mod-counting-formulae}]  
  Let $\Phi = \set{\phi_1, \ldots, \phi_\ell}$. For each $\bar n\in
  [0,p]^\ell$, let $\phi^{\Phi}_{\set{\bar n}}$ be given by
  Lemma~\ref{lem:counting-formulae} for $t:=p$, i.e. for each $\AS\in
  \fin_{\sigma,d}$, we have $\AS \models \phi^{\Phi}_{\set{\bar n}}$
  iff $\redmax{\bar n_\Phi(\AS)}{p}=\bar n$.
  Furthermore,
  $\qad{\phi^{\Phi}_{\set{\bar n}}} \leq \qad{\Phi} + 2$.  Below, for
  each $\bar r:=(r_1, \ldots, r_\ell)\in [0,p-1]^\ell$, $i\in [\ell]$,
  we construct a sentence $\chi^{\bar r}_i$ such that $\AS \models
  \chi^{\bar r}_i$ iff $n_{\phi_i}(\AS) \modeq{p} r_i$.  Furthermore,
  $\qad{\chi^{\bar r}_i} \leq \max\set{\qad{\Phi} + 1, 2(d-1)+2}$.  We
  can then define
  $\chi^{\Phi}_{R} \ :=\ \smashoperator{\biglor_{(\bar n, \bar r)\in R}} \big(\phi^{\Phi}_{\set{\bar n}} \land  \bigland_{i\in [\ell]} \chi^{\bar r}_{i}\big)$.
  Obviously, $\qad{\chi^{\Phi}_{R}} \leq \max\set{\qad{\Phi} + 2, 2(d-1)+2}$.
  
  Consider some $\bar r:=(r_1, \ldots, r_\ell)\in [0,p-1]^{\ell}$,
  $i\in [\ell]$, and let $\phi:=\phi_i$ and $r:=r_i$.  We define a
  formula $\phi^{=k}(x)$, such that $\AS\models \phi^{=k}(a)$, for
  $\AS\in \fin_{\sigma,d}$ and $a\in A$, iff $a$ belongs to a
  component $\KS$ of $\AS$ such that $\KS\models \phi$, $a\in
  \tdroot(\KS)$, and $\absval{\tdroot(\KS)}=k$.  Let $\tilde \phi(x)
  := \rela{\phi}{\reach_d(x,z)}$, let $\tilde
  \tdroot_d(x):=\rela{\tdroot_d(x)}{\reach_d(x,z)}(x)$, and let
  \begin{align*}
    \phi^{=k}(x):=\
    &\tilde \phi(x) \ \land\  \tilde{\tdroot}_d(x)\\[1ex]
    \land \ \exists &x_1 \ldots \exists x_k \ \Big( \smashoperator{\bigland_{j\in [k]}} \big(\tilde \tdroot_d(x_j) \land \reach_d(x_j,x) \ \land \ \smashoperator{\bigland_{j,j'\in [k],\, j\neq j'}} x_j \neq x_{j'}\big) \\[1ex]
    \land \ \forall &y\ \big(\tilde \tdroot_d(y) \land
    \smashoperator{\bigland_{j\in [k]}} y \neq x_j\big) \limplies
    \smashoperator{\bigland_{j\in [k]}} \lnot\reach_d(y,x) \Big).
  \end{align*}
  Observe that
  \begin{align*}
    \qad{\phi^{=k}} \ \leq \ &\max\set{\qad{\tilde \phi}, \qad{\tilde \tdroot_d} + 1, \qad{\reach_d} + 1}\\
    \leq \ &\max\set{\qad{\phi}, 2(d-1)+1}.
  \end{align*}
  Let the function $f$ be defined as in Lemma~\ref{lem:roots} and let
  $b:=f(d)$. Let $M\subseteq [0,p-1]^{b+1}$ be such that
  \begin{equation*}\label{eq:def-m}
      (a_0, \ldots, a_b)\in M \quad \text{ iff }\quad \sum_{k\in [0,b]}
  k \cdot a_k \modeq{p} r.
  \end{equation*}
  Now we define our formula
  $\chi^{\bar n}_i$ as 
  \[ \chi^{\bar n}_i  \ := \ \smashoperator{\biglor_{(a_1, \ldots, a_b)\in
      M}} \qquad \bigland_{k \in [0,b]}\exists^{k \cdot
    a_k\pmod*{p}}\, x \ \phi^{=k}(x)\, . \] 
  Obviously, $\qad{\chi^{\bar n}_i} \ \leq \ \max\set{\qad{\phi}, 2(d-1)+1} + 1$.

  We show that the formula is defined correctly. Let
  $\AS\in\fin_{\sigma,d}$. Recall that, according to
  Lemma~\ref{lem:roots}, $|\tdroot(\KS)| \leq b$ for each component
  $\KS$ of $\AS$.  We partition the set $H$ of components of $\AS$
  into pairwise disjoint sets $H_0, \ldots, H_b$ such that $\KS\in
  H_k$ iff $\absval{\tdroot(\KS)}=k$, for each $\KS\in H$. By
  definition of $\phi^{=k}(x)$, the number of elements $a\in A$ such
  that $\AS\models \phi^{=k}(a)$ equals $k \cdot |H_k|$. Hence, $\AS
  \models \chi^{\bar r}_{i}$ iff for some $(a_0, \ldots, a_b)\in M$,
  we have $k \cdot |H_k| \equiv k \cdot a_k \pmod*{p}$ for each $k\in
  [0,b]$. This is true iff $n_{\phi}(\AS) \equiv r \pmod*{p}$, since
  \[
  n_{\phi}(\AS) \ = \ \sum_{k\in [0,b]} k\cdot |H_k| \ \modeq{p} \
  \sum_{k\in [0,b]} k\cdot a_k  \modeq{p} r,
  \]
  for $a_0, \ldots, a_b\in [0,p-1]$ such that $|H_k| \modeq{p} a_k$
  for each $k\in [0,b]$.
\end{proof}

With these preparations, the proof of Theorem~\ref{thm:oimso-eq-fomod} is very
similar to the proof of Theorem~\ref{thm:oifo-eq-fo}.

\begin{proof}[Proof of Theorem~\ref{thm:oimso-eq-fomod}]
  The proof proceeds by induction on the tree-depth $d$.  We show that for each
  $\MSO[\sigma,\leqsym]$-sentence $\phi$ with $\qr{\phi} = q$, there is an
  $\FOmod[\sigma]$-sentence $\psi_{\phi,d}$ such that for each $\AS\in
  \fin_{\sigma,d}$, we have $\AS\models \psi_{\phi,d}$ iff
  $\tp^\leqsym_q(\AS)\models \phi$.  In particular, if $\phi$ is order-invariant,
  we let $\psi := \psi_{\phi,d}$, and we obtain $\AS \models_{\leqsym} \phi$ iff
  $\AS \models \psi := \psi_{\phi,d}$.
  
  Let $\ctypes[\sigma, q, d]=\set{\theta_{1}, \ldots, \theta_{\ell}}$.  We
  construct a sentence $\phi_i$ that defines $\theta_i$ on
  $\fin^\conn_{\sigma,d}$, for each $i\in [\ell]$. If $d=1$, the type of a
  connected structure of type $\theta_i$ is determined by the atomic $\sigma$-type
  $\alpha$ of its single element. We let $\phi^{\conn}_{\tau,1} \ := \ \exists x\,
  \alpha(x)$.  If $d>1$, for each $q$-type $\theta\in\types[\tilde \sigma, q,
  d-1]$, we obtain an $\FOmod$-sentence $\psi_{\theta,d-1}$ with
  $\qad{\psi_{\theta,d-1}} \leq 3(d-1)$.
  
  We construct $\phi_i$ according to Lemma~\ref{lem:connected-lift},
  i.e. we let $\phi_i:=\psi^{\conn}_{\theta_i,d}$ for each $i\leq
  \ell$. Let $\Phi := \set{\phi_1, \ldots, \phi_\ell}$.
  Note that $\qad{\Phi} \leq 3(d-1)+2$. 

  Now consider a sentence $\phi\in\MSO[\sigma,\leqsym]$.
  Let
  \[ R:=\lrsetc{\left(\redmax{\bar{n}_{\types[\sigma,q]}(\BS)}{p},
    \redmod{\bar{n}_{\types[\sigma,q]}(\BS)}{p}\right)}{\BS \in \fin_{\sigma,d},\tp^\leqsym_q(\BS) \models\phi} \] 
  where $p$ is given by the Pumping Lemma for $q$.
  We construct $\psi_{\phi,d}:=\psi^{\Phi}_{R}$ according to Lemma~\ref{lem:mod-counting-formulae}.
  In particular, $\qad{\psi_{\phi,d}} \ \leq \ \qad{\Phi} + 1 \ \leq 3d$.
  Consider some $\AS\in \fin_{\sigma,d}$.
  Observe that, for each component $\KS$ of $\AS$, we have $\KS \models \phi_{i}$ iff $\tp^\leqsym_q(\KS)=\tau_i$. Hence,
  $(\redmax{\bar{n}_{\Phi}(\AS)}{p},\redmod{\bar{n}_{\Phi}(\AS)}{p}) = (\redmax{\bar{n}_{\types[\sigma,q]}(\AS)}{p},\redmod{\bar{n}_{\types[\sigma,q]}(\AS)}{p})$.
  Thus
  \[ \AS \models \psi_{\phi,d} \ \iff \ (\redmax{\bar{n}_{\types[\sigma,q]}(\AS)}{p},\redmod{\bar{n}_{\types[\sigma,q]}(\AS)}{p}) = (\redmax{\bar{n}_{\types[\sigma,q]}(\BS)}{p},\redmod{\bar{n}_{\types[\sigma,q]}(\BS)}{p}) \] for some structure $\BS\in \fin_{\sigma,d}$ with $\tp^\leqsym_q(\BS)\models \phi.$
  As a consequence of Lemma~\ref{lem:oimso-cut}, this holds iff
  $\tp^\leqsym_q(\AS)\models \phi$.
\end{proof}

\section{Monadic second-order logic}
\label{sec:mso}

In \cite{ElberfeldGT12} it was proved that each $\MSO{}$-definable class of
finite graphs of bounded tree-depth is also $\FO$-definable.  Our approach
towards the results of the previous section can be adapted to obtain another
proof of this result which allows us to give an elementary upper bound on the
size of the $\FO$-sentence in terms of the quantifier-rank of the
$\MSO{}$-sentence.  Throughout this section, we assume in all notation whose
definition refers to a logic $\logl$ that $\logl=\MSO{}$.  We let
$\types:=\setc{\tp_q(\AS)}{\AS\in\fin_{\sigma,d}}$ and let
$\ctypes:=\setc{\tp_q(\AS)}{\AS\in\fin^\conn_{\sigma,d}}$.

\begin{theorem}
  \label{thm:mso-eq-fo}
  Let $d\in \Npos$ and let $\sigma$ be a signature. For each
  $\MSO[\sigma]$-sentence $\phi$ there is an $\FO[\sigma]$-sen\-tence $\psi$ with
  $\size{\psi}\in \nexp[d](\qr{\phi})$ and $\qad{\psi} \leq 2d$ that is equivalent
  to $\phi$ on $\fin_{\sigma,d}$.
\end{theorem}

We also prove the following theorem in Section~\ref{sec:lower-bounds} below
which shows that the upper bound of Theorem~\ref{thm:mso-eq-fo} is essentially
optimal.

\begin{theorem}
  \label{thm:mso-lower-bound}
  There is a signature $\sigma$ such that for each $d\in \Npos$ there is an
  $\MSO[\sigma]$-sentence $\phi_d$ such that each $\FO[\sigma]$-sentence $\psi_d$
  that is $\fin_{\sigma,d}$-equivalent to $\phi_d$ has size $\size{\psi_d} \geq
  \nexp[\size{\phi_d}](0)$.
\end{theorem}

\subsection{From MSO to FO}

Much of the proof of Theorem~\ref{thm:mso-eq-fo} follows the proof of
Theorem~\ref{thm:oifo-eq-fo}, but we are spared of the complications that arose
in connection with the ordering of structures.  Overall, this makes the proof of
Theorem~\ref{thm:mso-eq-fo} simpler.  On the other hand, the proof of an
analogue to Lemma~\ref{lem:cut-determines-type} becomes somewhat more
complicated.

\paragraph{Counting components}

In Lemma~\ref{lem:cut-determines-type}, we did not use the fact that we consider
only structures of bounded tree-depth. Here naively ignoring the bounded
tree-depth would cause the component counting threshold for $\MSO$-sentences of
quantifier-rank $q$ to depend non-elementarily on $q$. We use the following
lemma to avoid this.

\begin{lemma}
  \label{lem:msocut}
  Let $d,q\in \Npos$. There is a $t:=t(d,q)\in \nexp[d](q)$ such that
  for all structures $\AS, \BS\in \fin_{\sigma,d}$,
  \[
  \bar n_{\types[\sigma,q]}(\AS) \maxeq{t} \bar n_{\types[\sigma,q]}(\BS)
  \quad\implies\quad
  \AS \msoeleq[q] \BS.
  \]
\end{lemma}

Lemma~\ref{lem:msocut} is an easy consequence of the following two lemmas.

\begin{lemma}
  \label{lem:mso-cut-bounded-size}
  Let $k\in \Npos$, $q\in \N$, and $t:=2^{kq}$. Let $\sigma$ be a
  signature. For all structures $\AS,\BS\in\fin_\sigma$ whose
  components each contain at most $k$ elements,
  \[
  \bar n_{\types[\sigma,q]}(\AS) \maxeq{t} \bar n_{\types[\sigma,q]}(\BS) 
  \quad\implies\quad
  \AS \msoeleq[q] \BS.
  \]
\end{lemma}

\begin{lemma}
  \label{lem:component-bound}
  Let $d,q\in \Npos$ and let $\sigma$ be a signature. Each structure
  $\AS\in\fin_{\sigma,d}$ contains an induced substructure $\BS$ with
  $|B|\in \nexp[d](q)$ and $\AS \msoeleq[q] \BS$. If $\AS$ is
  connected, there is such a structure $\BS$ with $|B|\in
  \nexp[(d-1)](q)$.
\end{lemma}

Before we prove Lemma~\ref{lem:mso-cut-bounded-size} and
Lemma~\ref{lem:component-bound}, we show how to prove Lemma~\ref{lem:msocut}
with their help. The proof will also use the following variant of a standard
composition lemma, which we take for granted (we use a variant for signatures
with constants, where the constant symbols will be used in the proof of
Lemma~\ref{lem:mso-cut-bounded-size}).

The definition of the disjoint union $\AS \disunion \BS$ of structures $\AS$ and
$\BS$ can be extended to signatures with constant symbols, if the constant
symbols of $\AS$ and $\BS$ are disjoint.

\begin{lemma}[Composition Lemma]
  Let $q\in \N$. Let $\sigma_1, \sigma_2$ be signatures which may
  contain constant symbols, where the constants in $\sigma_1$ and
  $\sigma_2$ are disjoint. If $\AS_1,\BS_1$ are $\sigma_1$-structures
  and $\AS_2,\BS_2$ are $\sigma_2$-structures such that
  $\AS_1\msoeleq[q] \BS_1$ and $\AS_2\msoeleq[q] \BS_2$, then
  \[ \AS_1 \disunion \AS_2 \msoeleq[q] \BS_1 \disunion \BS_2. \]
\end{lemma}

\begin{proof}[Proof of Lemma~\ref{lem:msocut}]
  With the help of Lemma~\ref{lem:component-bound} and the Composition
  Lemma, we can assume without loss of generality that $\AS$ and $\BS$
  contain only components of size at most $k\in \nexp[(d-1)](q)$. Let
  $t:=2^{kq}$ as in Lemma~\ref{lem:mso-cut-bounded-size}. Then $t\in
  \nexp[d](q)$ and hence the claim follows from
  Lemma~\ref{lem:mso-cut-bounded-size}.
\end{proof}

\begin{proof}[Proof of Lemma~\ref{lem:mso-cut-bounded-size}]
  For the proof, we consider signatures $\sigma$ which contain
  constant symbols. In this case, the components of a
  $\sigma$-structure are not necessarily $\sigma$-structures, because
  they might not contain all constants. Let $T_{\sigma,q}$ denote the
  union of the sets of $(\MSO,\sigma',q)$-types over all signatures $\sigma' \subseteq \sigma$. For $\sigma$-structures $\AS,\BS$
  and $q,t\in \Npos$, we write $\AS \approx_{q,t} \BS$ if $\bar
  n_{T_{\sigma,q}} \maxeq{t} \bar n_{T_{\sigma,q}}$.

  By induction on $q$, we prove the stronger claim that for each
  signature $\sigma$ \emph{which may contain constant symbols} and all
  $\sigma$-structures $\AS$ and $\BS$ whose components each contain at
  most $k$ elements,
  \[ \AS \approx_{q,t} \BS \ \implies \ \AS \msoeleq[q] \BS. \] Let
  $q=0$. Since $\AS \approx_{q,1} \BS$, there exists a bijection $f$
  between the sets $M_\AS, M_\BS$ of components of $\AS,\BS$ which
  contain constants. Furthermore, this bijection preserves the
  $0$-type of components, i.e. for each component $\KS\in M_\AS$ there
  exists a partial isomorphism $g_\KS$ whose domain and codomain are,
  respectively, the set of constants of $\KS$ and $f(\KS)$. These
  partial isomorphisms can be extended to a partial isomorphism $g :=
  \bigunion_{\KS\in M_\AS} g_\KS$ of $\AS$ and $\BS$ whose domain and
  codomain are, respectively, the set of constants of $\AS$ and
  $\BS$. Hence $\AS \msoeleq[0] \BS$.

  For each $q\in \N$, let $t(q):=2^{kq}$. Now let $q>0$. We consider
  the case where $\AS$ and $\BS$ contain only components of a single
  $q$-type $\tau$ over some signature $\sigma'\subseteq \sigma$. The
  general case follows by an application of the Composition Lemma. By
  a further application of the Composition Lemma, we can assume that
  all components of $\AS$ and $\BS$ are isomorphic to a single
  structure $\KS$ of type $\tau$. Now if $n_\tau(\AS)=n_\tau(\BS)$,
  then $\AS$ and $\BS$ are isomorphic, so we are done. Assume that
  $n_\tau(\AS), n_\tau(\BS) > t(q)$. We show that Duplicator wins the
  $q$-round EF-game on $\AS$ and $\BS$.
  
  Consider the first round of the game. Suppose that Spoiler plays a
  point move, i.e. he chooses an element, say, $a\in A$. Duplicator
  chooses an element $b$ corresponding to $a$ in a copy of $\KS$ in
  $\BS$. This introduces exactly one component of a new
  isomorphism-type $\tau'$ in each of $(\AS,a)$ and $(\BS,b)$. The
  remaining components of $(\AS,a)$, $(\BS,b)$ all remain their
  isomorphism-type and there are more than $t(q)-1 \geq t(q-1)$ such
  components. Hence $(\AS,a) \approx_{q-1,t(q-1)} (\BS,b)$. By
  induction, $(\AS,a) \msoeleq[q-1] (\BS,b)$. So Duplicator wins, if
  she replies by $b$.

  Suppose now that Duplicator plays a set move, say, $M\subseteq A$.
  Since $\KS$ contains at most $k$ elements, the components of the
  structure $(\AS,M)$ belong to at most $2^k$ different
  isomorphism-types. Thus the number of $q$-types cannot be greater
  either. For each $q$-type $\theta$ occurring in $(\AS,M)$, let
  $C_\theta$ denote the set of components of $\AS$ whose $q$-type is
  $\theta$. Duplicator chooses a set $C'_\theta$ of components of
  $\BS$ and a set of elements $M'_\theta\subseteq \bigunion_{\CS\in
    C'_\theta} C$ such that $\min\set{\absval{C_{\theta}},
    t(q-1)}=\min\set{\absval{C'_{\theta}}, t(q-1)}$, and $\tp_q(\CS,M'_\theta
  \intersect C)=\theta$ for each $\CS \in C'_\theta$. Since there are $t(q) > 2^k \cdot t(q-1)$
  copies of $\KS$ in $\BS$, this is possible. Let $M' :=
  \bigunion_{\theta} M'_\theta$. We have $(\AS,M) \approx_{q-1,t(q-1)}
  (\BS,M')$. So, by induction, $(\AS,M) \msoeleq[q-1]
  (\BS,M')$. Replying by $M'$, Duplicator wins.
\end{proof}

Lemma~\ref{lem:component-bound} is an adaptation of
\cite[Thm. 6.7]{NesetrilMendez2012} from \FO{} to \MSO{}.  Its proof uses the
previous lemma and the following analogue to Lemma~\ref{lem:IndASrtoAS}, which
can be proved like Lemma~\ref{lem:IndASrtoAS}.

\begin{lemma}
  \label{lem:mso-IndASrtoAS}
  Let $q\in \Npos$. Let
  $\AS,\BS\in \fin_\sigma$ be connected structures with $\td(\AS), \td(\BS) > 1$
  and let $r_\AS\in \tdroot(\AS)$,$r_\BS\in \tdroot(\BS)$
  with $\alpha(\AS,r_\AS)=\alpha(\BS,r_\BS)$.
  Then \[ \AS^{[r_\AS]} \msoeleq[q] \BS^{[r_\BS]} \ \implies \ \AS \msoeleq[q] \BS. \]
\end{lemma}
\begin{proof}[Proof of Lemma~\ref{lem:component-bound}]
  The proof is by induction on the tree-depth $d$.  First, we consider the claim
  about connected structures.  If $d=1$, then each connected structure with
  $\td(\AS)=1$ has size $1\in \nexp[0](q)$, i.e. we can set $\BS:=\AS$.  Suppose
  now that $d>1$. Choose a tree-depth root $r\in \tdroot(\AS)$.  By induction,
  since $\td(\AS^{[r]}) \leq d-1$, we obtain an induced substructure $\BS'$ of
  $\AS^{[r]}$ such that $|B'|\in \nexp[(d-1)](q)$ and $\BS' \msoeleq[q]
  \AS^{[r]}$.  Let $\BS$ be the substructure of $\AS$ induced by $B' \union
  \set{r}$, i.e. $\BS^{[r]}=\BS'$. Since $\AS^{[r]} \msoeleq[q] \BS^{[r]}$, we
  obtain that $\AS \msoeleq[q] \BS$ in the same way as in
  Lemma~\ref{lem:IndASrtoAS}. Observe that $|B|\in \nexp[(d-1)](q)$.

  Consider the case that $\AS$ is not connected. By the construction
  above, we can replace each component $\KS$ of $\AS$ by an induced
  substructure of $\KS$ on $\nexp[(d-1)](q)$ vertices that has the
  same $q$-type as $\KS$. By the Composition Lemma, this preserves the
  $q$-type of $\AS$. Let $k\in \nexp[(d-1)](q)$ denote the maximum
  number of vertices in a component of $\AS$ after this
  replacement. By Lemma~\ref{lem:mso-cut-bounded-size}, we know that
  $\BS \msoeleq[q] \AS$ for each induced substructure $\BS$ of $\AS$
  such that $n_{\tau}(\BS) \maxeq{t} n_{\tau}(\AS)$ for each $q$-type
  $\tau$, where $t:=2^{kq}$. Since there are at most $2^k$
  non-isomorphic components in $\AS$ and we have to keep at most $t$
  copies of each such component, there is such a structure $\BS$ with
  $|B|\in \nexp[d](q)$.
\end{proof}

\paragraph{Finishing the proof}

With the preparations above, the proof of Theorem~\ref{thm:mso-eq-fo} is now
very similar to the proof of Theorem~\ref{thm:oifo-eq-fo}. 

\begin{proof}[Proof of Theorem~\ref{thm:mso-eq-fo}]
  The proof proceeds by induction on the tree-depth $d$, where we also show that
  $\absval{\types} \in \nexp[d](q)$ and $\absval{\ctypes} \in \nexp[(d-1)](q)$. 

  \paragraph{Defining types of connected structures}

  As a first step, we prove that each $q$-type $\tau\in \ctypes$ is
  $\fin_{\sigma,d}^{\conn}$-equivalent to an $\FO[\sigma]$-sentence
  $\phi^{\conn}_{\tau,d}$ such that $\size{\phi^{\conn}_{\tau,d}}\in
  \nexp[(d-1)](q)$ and $\qad{\phi^{\conn}_{\tau,d}} \leq 3(d-1)+1$.  For $d=1$,
  each structure $\AS\in\fin_{\sigma,d}^{\conn}$ of type $\tau$ consists of a
  single element of some atomic $\sigma$-type $\alpha$. The $\FO{}$-sentence
  $\phi^{\conn}_{\tau,1} := \exists x\, \alpha(x)$ then defines $\tau$. Hence
  $\size{\phi^{\conn}_{\tau,1}}$ does not depend on $q$,
  $\qad{\phi^{\conn}_{\tau,1}} = 0$, and $\absval{\ctypes} \leq \nexp[0](q)$.

  Now suppose that $d>1$ and let $\tau\in \ctypes[\sigma,q,d]$.  Let $R\subseteq
  \types[\tilde \sigma,q,d-1] \times 2^\sigma$ be a set that contains
  $(\theta,\alpha)$ iff there is a structure $\BS\in \fin_{\tilde
    \sigma,d}^{\conn}$ with $\tp_{q}(\BS)=\tau$ which contains a tree-depth root
  $r\in \tdroot(\BS)$ such that $\alpha(\BS,r)=\alpha$ and
  $\tp_{q}(\BS^{[r]})=\theta$.  Observe that, as a consequence of Lemma
  \ref{lem:mso-IndASrtoAS}, for each $\AS\in \fin^\conn_{\sigma,d}$, we have
  $\tp_q(\AS)=\tau$ iff $(\tp_q(\AS^{[r]}), \alpha(\AS,r))\in R$ for some $r\in
  \tdroot(\AS)$.  Now consider a $q$-type $\theta\in \ctypes[\tilde \sigma,q,d-1]$
  and let $\phi_{\theta,d-1}$ be the $\FO{}[\tilde \sigma]$-sentence, given by
  induction, which is equivalent to $\theta$ on $\fin_{\tilde
    \sigma,d-1}^{\conn}$.  As a consequence of Lemma~\ref{lem:ASrtoAS}, we obtain
  that for all structures $\AS\in\fin^\conn_{\sigma,d}$ with $\td(\AS) > 1$ and
  all tree-depth roots $r\in \tdroot(\AS)$, we have
  $\AS\models\I(\phi_{\theta,d-1})(r)$ iff $\tp_q(\AS^{[r]})=\theta$.

  Altogether, we obtain that the following $\FO[\sigma]$-sentence is equivalent to $\tau$ on $\fin^\conn_{\sigma,d}$:
  \[ \phi^{\conn}_{\tau,d}\ :=\ (\td\leq 1 \land \phi_{\tau,d-1})
  \ \lor \ \smashoperator{\biglor_{(\theta,\alpha) \in
      R}}\exists x\ \big(\tdroot_d(x) \ \land \ \alpha(x)\
  \land\ \I(\phi_{\theta,d-1})(x)\big)\, .\]
  
  Recall that, by induction, $\size{\I(\phi_{\theta,d-1})}\in \nexp[(d-1)](q)$
  and $\absval{\types[\tilde \sigma, q,d-1]}\in \nexp[(d-1)](q)$.  Hence,
  $\absval{R} \in \nexp[(d-1)](q)$.  Altogether, we obtain that
  $\size{\phi^{\conn}_{\tau,d}} \in \nexp[(d-1)](q)$.  Using
  Lemma~\ref{lem:mso-IndASrtoAS}, we conclude that
  $\absval{\ctypes[\sigma,q,d]}\leq 2^\sigma \cdot \absval{\types[\tilde
    \sigma,q,d-1]} \in \nexp[(d-1)](q)$. By induction, $\qad{\I(\phi_{\theta,d-1})}
  \leq 3(d-1)$. Hence, $\qad{\phi^{\conn}_{\tau,d}} \leq 3(d-1)+1$.

  \medskip

  \paragraph{Structures with multiple components}
  Consider an $\MSO[\sigma]$-sentence $\phi$.  Let
  $\ctypes:=\set{\tau_{1}, \ldots, \tau_{\ell}}$, where
  $\ell:=|\ctypes|$. Let $t:=t(d,q)\in \nexp[d](q)$ be given by
  Lemma~\ref{lem:msocut}.  Let $\Phi$ be the set that contains the
  formulae $\phi_{i}:=\phi^{\conn}_{d,\tau_{i}}$ for each
  $i\in[\ell]$. Hence, $\bar{n}_{\Phi}(\AS) \maxeq{t}
  \bar{n}_{\ctypes}(\AS)=\bar{n}_{\types[\sigma,q]}(\AS)$ for each $\AS\in \fin_{\sigma,d}$.  Let $R\subseteq
  [0,t]^{\ell}$ be a set such that $\bar{n}\in R$ iff there exists a
  model $\AS\in\fin_{\sigma,d}$ of $\phi$ with
  $\redmax{\bar{n}_{\Phi}(\AS)}{t}=\bar{n}$.  Using
  Lemma~\ref{lem:msocut}, we obtain that $\AS\models \phi$ iff
  $\redmax{\bar n_{\Phi}(\AS)}{t}\in R$, for each
  $\AS\in\fin_{\sigma,d}$. Hence, the $\FO[\sigma]$-sentence
  $\psi:=\psi^\Phi_R$ of Lemma~\ref{lem:counting-formulae} is
  equivalent to $\psi$ on $\fin_{\sigma,d}$.

  Regarding the size of $\psi$, note that Lemma~\ref{lem:msocut} implies that
  $\absval{R} \ \leq \ \absval{\ctypes} \ \leq \ [0,t]^{\ell}$.
  Since
  \begin{align*}
    t^\ell \ \in \ (\nexp[d](q))^{\nexp[(d-1)](q)} \ = \ &(2^{\nexp[(d-1)](q)})^{\nexp[(d-1)](q)}
    \\ = \ &2^{\nexp[(d-1)](q) \cdot \nexp[(d-1)](q)}\\
    \subseteq\ &2^{\nexp[(d-1)](q)} \ = \ \nexp[d](q)
  \end{align*}
  we obtain that, by the construction of $\psi$ according to Lemma~\ref{lem:counting-formulae},
  \begin{align*}
    \size{\psi} \ \leq \ &c \cdot \absval{\Phi} \cdot \size{\Phi} \cdot \absval{R} \cdot t^2.\\
    \in\ &\nexp[(d-1)](q) \cdot \nexp[d](q) \cdot \nexp[d](q)^2 \cdot \nexp[(d-1)](q)\\
    \subseteq\ &\nexp[d](q),
  \end{align*}
  and $\qad{\psi} \leq \qad{\Phi} + 2 \leq 3d$.
\end{proof}

\subsection{A lower bound}
\label{sec:lower-bounds}

The proof of Theorem~\ref{thm:mso-lower-bound} uses an encoding of large natural
numbers $n$ by shallow trees $\enc(n)$ from \cite[chapter 10.3]{FlumGrohe2004}.
Here, by \emph{trees}, we mean directed trees which are rooted, i.e. trees which
contain a root vertex from which all edges point away.  The encoding is defined
inductively as follows:
\begin{itemize}
\item $\enc(0)$ is the one-node tree.
\item For $n\geq 1$, the tree $\enc(n)$ is obtained by creating a new
  root and attaching to it all trees $\enc(i)$ such that the $i$-th
  bit in the binary representation of $n$ is $1$.
\end{itemize}
Note that a tree encodes a number with respect to this encoding iff there are no
two distinct isomorphic subtrees whose roots are children of the same
vertex. But we would like to assign a natural number to each tree. To this end,
we reduce each tree $\TS$ in a bottom-up way to a tree $\num(\TS)$ that encodes
a number:
\begin{itemize}
\item $\num(\TS):=\TS$ if $\height(\TS)=1$, i.e. $\TS\iso \enc(0)$.
\item If $\height(\TS) > 1$, select one tree $\TS_1, \ldots, \TS_k$ of each
  isomorphism type that occurs among the immediate subtrees of the root of
  $\TS$. Define $\num(\TS)$ to be a tree whose root has children whose rooted
  subtrees are $\num(\TS_1), \ldots, \num(\TS_k)$.  
\end{itemize}
Throughout the following section, we let $\sigma:=\set{E,R,B}$, where $E$ is a
binary and $R,B$ are unary relation symbols.  We consider a tree as a
$\set{E}$-structure $\TS$ where $E^\TS$ is the edge relation of the tree.  A
\emph{coloured tree} is a finite $\sigma$-structure $(\TS,R^\TS,B^\TS) $, where
$\TS$ is a tree and $R^\TS,B^\TS$ (the \emph{red} and the \emph{blue} vertices
of $\TS$) form a partition of the vertex set of the tree.  Structures whose
components are (coloured) trees are called (coloured) \emph{forests}.  The
\emph{height} $\height(\TS)$ of a (coloured) tree $\TS$ is the maximum number of
vertices on a path from the root of $\TS$ to a leave of $\TS$.  The
\emph{height} $\height(\FS)$ of a (coloured) forest $\FS$ is the maximum height
of its components.

From the proof of \cite[Lemma~10.21]{FlumGrohe2004}\footnote{\cite[Lemma~10.21]{FlumGrohe2004}
  makes the assumption that $\TS_{1},\TS_{2}$ are encodings of
  numbers $n,m$ to conclude that $\FS\models \text{eq}_d(u_1,u_2) \
  \iff \ n=m$, i.e. $\TS_{1} \iso \TS_{2}$. If we drop this assumption, we
  obtain our variant of the lemma using exactly the same formula.}, we
obtain the following lemma.

\begin{lemma}
  \label{lem:eq-formula}
  For each $d\in\Npos$, there is an $\FO[E]$-formula $\text{eq}_d(x,y)$
  of size $\size{\text{eq}_d}\in \bigo(d)$ such that for all forests
  $\FS$ with $\height(\FS)\leq d$ and all trees $\TS_1, \TS_2$ of $\FS$ with roots
  $u_1,u_2$, respectively, we have:
  \[ \FS\models \text{eq}_d(u_1,u_2) \ \iff \ \num(\TS_1)=\num(\TS_2). \]
\end{lemma}

Note that $\height(\enc(n))\leq d$ provided that $n<\tower(d)$, where $\tower(d) := \nexp[d](0)$.
For each $d\geq 1$, let $\FS_d$ denote a coloured forest that contains exactly the trees $\enc(0), \ldots,
\enc(\tower(d)-1)$ whose vertices all are coloured red, let $\TS_d$ denote a
coloured tree with $\height(\TS_d) \leq d$ that contains each of the trees
$\enc(0), \ldots, \enc(\tower(d-1)-1)$? as subtrees (e.g. a full $\tower(d-1)$-ary
tree) and where all vertices are blue, and let $\FS_d^n$ denote the disjoint
union of $\FS_d$ and $n$ disjoint copies of $\TS_d$, for each $n\geq 0$.

\begin{lemma}
  \label{lem:mso-lower-bound-main-lemma}
  For each $d\in \Npos$, there exists an $\MSO[\sigma]$-sentence
  $\phi_d$ of size $\bigo(d)$ such that $\FS_d^n \models \phi_d$ iff $n
  \geq \tower(d)$.
\end{lemma}
\begin{proof}
  Let $d\in \Npos$ and let $\text{eq}_d(x,y,M)$ be the
  relativisation of the $\FO[E]$-formula of Lemma~\ref{lem:eq-formula}
  to a set variable $M$. Let $\text{conn}(M)$ be an $\MSO[E]$-formula
  which states in a forest $\FS$ that for each tree $\TS$ of $\FS$,
  the structure induced by $M$ in $\TS$ is connected, i.e. a tree.
  Let $\text{root}(x,M)$ state that $x$ is a root in the
  subforest induced by $M$.
  We can assume that the size of $\text{conn}(M)$ and $\text{root}(x,M)$ is independent
  of $d$.
  Now let $\phi_d$ be the following sentence:
  \begin{align*}
    \exists M\ \Big(\conn(M) \ \land \ \forall x\ \big(&R(x) \land \text{root}(x,M)\big)  \limplies\\
    &\exists y\ \big(\text{root}(y,M) \ \land \ B(y) \ \land \
    \text{eq}_h(x,y,M)\big)\Big).
  \end{align*}
  First we argue that $n \geq \tower(h)$ implies $\FS^n_d\models \phi_d$.
  By definition, the red trees contained in $\FS_d^n$ are $\enc(0),
  \ldots, \enc(\tower(d)-1)$. Since $n\geq \tower(h)$, we can choose
  $\tower(h)$ pairwise distinct copies $\HS_0, \ldots, \HS_{\tower(h)-1}$ of $\TS_d$
  in $\FS_d^n$. Since all trees $\enc(0), \ldots, \enc(\tower(d)-1)$
  occur as subtrees of $\TS_d$, for each $i\in [0,\tower(d)-1]$
  there is a set $M_i\subseteq H_i$ with $\rela{(\HS_n[M_n])}{E}\iso
  \enc(i)$.  The set $M:=M_1 \union \dotsb \union M_n$ witnesses that $\FS_d^n\models
  \phi_d$.
  
  Now we show that $\FS_d^n\models \phi_d$ implies $n \geq \tower(h)$.
  Let $M\subseteq F_d^n$ witness that
  $\FS_d^n\models \phi_d$. The forest $\FS_d^n$ contains trees
  $\enc(0),\ldots,\enc(\tower(d)-1)$ whose vertices are all red. Hence, and according to the
  choice of $M$ and the choice of $\text{eq}_h(x,y,M)$, for
  each $i\in [0,\tower(d)-1]$ there is a
  blue copy $\TS$ of $\TS_d$ in $\FS_d^n$ such that
  $\num(\TS[M])=\num(\enc(i))=i$. Hence $\FS_d^n$ must contain at least
  $\tower(h)$ copies of $\TS_d$, because $M$ induces at most one tree in
  each copy of $T_h$.
\end{proof}

Using Lemma~\ref{lem:mso-lower-bound-main-lemma}, we can easily finish the proof
of Theorem~\ref{thm:mso-lower-bound}.

\begin{proof}[Proof of Theorem~\ref{thm:mso-lower-bound}]  
  $\FO$-sentences of quantifier-rank $q$ cannot distinguish $\FS_d^k$
  from $\FS_d^{k+1}$ for each $k\geq q$. Hence an $\FO$-sentence
  $\psi_d$ that is equivalent to the $\MSO$-sentence $\phi_d$ of
  Lemma~\ref{lem:mso-lower-bound-main-lemma} must have quantifier-rank
  $\qr{\psi_d}\geq \tower(d)$ and in particular $\size{\psi_d} \geq
  \tower(d)$.
\end{proof}

\section{Defining Bounded-Depth Tree-Decompositions in FO}
\label{sec:canondecomp}

For every finite relational signature $\sigma$ and every $k \in \N$
there is a set $\Sigma(\sigma,k)$ of labels such that information
about a $\sigma$-structure $\AS$ of tree-width at most $k$ may be
encoded into a $\Sigma(\sigma,k)$-labelled tree $T_\AS$. This encoding
may be chosen so that the original structure $\AS$ can be interpreted in
$T_\AS$ by an $\MSO$-interpretation. One such encoding is presented in
details in~\cite[Section~11.4]{FlumGrohe2004}.\footnote{That $\AS$ can
  be $\MSO$-interpreted in $T_\AS$ is not proved there but easy to see.}

The question of whether there is an interpretation in the converse
direction, i.e. whether some tree $T_\AS$ representing a width-$k$
tree-decomposition of $\AS$ can be $\MSO{}$-interpreted in $\AS$, is still
open. In particular, interpretability of such a decomposition would
imply that recognisability equals $\CMSO$-definability for graphs of
bounded tree-width.

In this section we show that for graphs of bounded tree-\emph{depth},
there is even an $\FO{}$-interpretation of a bounded-depth
tree-decomposition. Since the interpretation we give here is not
parameterised we obtain a canonical tree-decomposition, though not one
of optimal depth or width. The $\FO{}$-interpretation is given by
formulae $\epsilon_d(x,y)$ and $\alpha_d(x,y)$ for every $d \geq 1$
such that if $\AS$ is a $\sigma$-structure of tree-depth at most $d$
then
\begin{itemize}
\item $\epsilon_d$ defines an equivalence relation $\sim_\AS := \{ (u,v) \st \AS
  \models \epsilon_d[u,v] \}$ on $A$, 
\item the equivalence classes of $\sim_\AS$ have size bounded by a
  function of $d$,
\item the relation defined by $\alpha_d$ is invariant under $\sim_\AS$,
  i.e. if $u \sim_\AS u'$ and $v \sim_\AS v'$, then
  \[
  \AS \models \alpha_d(u,v) \iff \AS \models
  \alpha_d(u',v'),\text{ and}
  \]
\item $\alpha_d$ defines a rooted tree structure on the quotient structure $\AS/\!\!\!\sim_{\AS}$,
  in which $[u]_{\sim_\AS}$ is an ancestor of $[v]_{\sim_\AS}$ or vice
  versa whenever $u,v \in A$ are adjacent in the Gaifman graph of
  $\AS$.
\end{itemize}
This can be turned into a bounded-depth tree-decomposition in the
usual sense by taking the tree structure on $\AS/\!\!\sim_\AS$ as the tree
and setting 
$\{ v \st [v]_{\sim_\AS}\text{ is an ancestor of }[u]_{\sim_\AS} \}$
as the bag of the node $[u]_{\sim_\AS}$.

The key insight we use is Lemma~\ref{lem:roots} which says that for any fixed
$d$ there are at most $f(d)$ many candidates which may be placed at the root of
a tree-decomposition of $\AS$ of minimum height. We have already seen at the end
of Section~\ref{sec:background} that there is an $\FO{}$-formula
$\tdroot_d(x)$ such that $\AS \models \tdroot_d[r]$ iff $r$ is such a
candidate. We recursively build a tree-decomposition $\mathcal{T}_\AS$ of $\AS$ of
height at most $d$ by placing, in each step, \emph{all} candidate roots into the
root-bag of our tree-decomposition and then recursing on the components of the
remaining graph. Note that even if $\td(\AS) = d$, not all components of $\AS \setminus R$,
where $R$ is the set of at most $f(d)$ root nodes, necessarily have tree-depth
$d-1$, so we must be a bit careful which elements we place into the root of the
next level.

We fix a tree-depth $d$ and recursively define \FO-formulae
$\varphi_i$ for $i = 0,\ldots,d$ with the intended meaning that,
in a structure $\AS$ of tree-depth $d$ with $a \in A$,
$\AS \models \varphi_i[a]$ iff $a$ is on the $i$-th level of the
tree-decomposition, which we denote by $L_{i}$: 
\[\begin{split}
\varphi_0(x) &:= \bot
\\
\varphi_i(x) &:= 
\bigvee_{j=1}^{d-i} \left(\rela{\td_{=j+1}}{\neg\varphi_{<i}} \wedge
 \rela{\td_{=j}}{\neg(\varphi_{<i} \vee z \dot = x)} \right)
\end{split}
\]
Here, $x$ is the free variable of $\varphi_i$ and $z$ is the free
variable of the formulae used in the restrictions. With the
abbreviations
\[
\varphi_{< i}(x) := \bigvee_{j < i}
\varphi_j(x)\quad\text{and}\quad
\varphi_{\leq i}(x) := \bigvee_{j \leq i}
\varphi_j(x)
\]
we define
\[\begin{split}
\psi_0(x,y) &:= \top
\\
\psi_{i+1}(x,y) &:=
\rela{\reach_{d-i+1}}{\neg \varphi_{\leq i}},
\end{split}
\]
i.e. $\psi_i(u,v)$ holds iff $u$ and $v$ are in the same connected
component of $\AS \setminus \bigcup_{j \leq i} L_j$. We can now define an
equivalence relation on $\AS$ as follows:
\[
\epsilon_d(x,y) :=
\bigvee_{1 \leq i \leq d}
(\varphi_i(x) \wedge \varphi_i(y) \wedge \psi_i(x,y)),
\]
i.e. two elements are equivalent iff they appear on the same level of
our tree-decomposition and are in the same connected component of $\AS$
after removing the levels above $x$ and $y$. This is equivalent to
saying that $x$ and $y$ appear in the same node of our
tree-decomposition.

Let $\gamma(x,y)$ be a formula which expresses
that to elements are adjacent in the Gaifman graph of a structure.
Finally, We define tree edges (directed towards the root) by
\[
\alpha_d(x,y) :=
\bigvee_{1 \leq i < d}
(\varphi_i(x) \wedge \varphi_{i+1}(y) \wedge
\exists u\exists v\,(\gamma(u,v) \wedge \epsilon(x,u) \wedge \psi_{i+1}(y,v))).
\]

\begin{figure}[htb]
\begin{center}
  \resizebox{0.9\textwidth}{!}{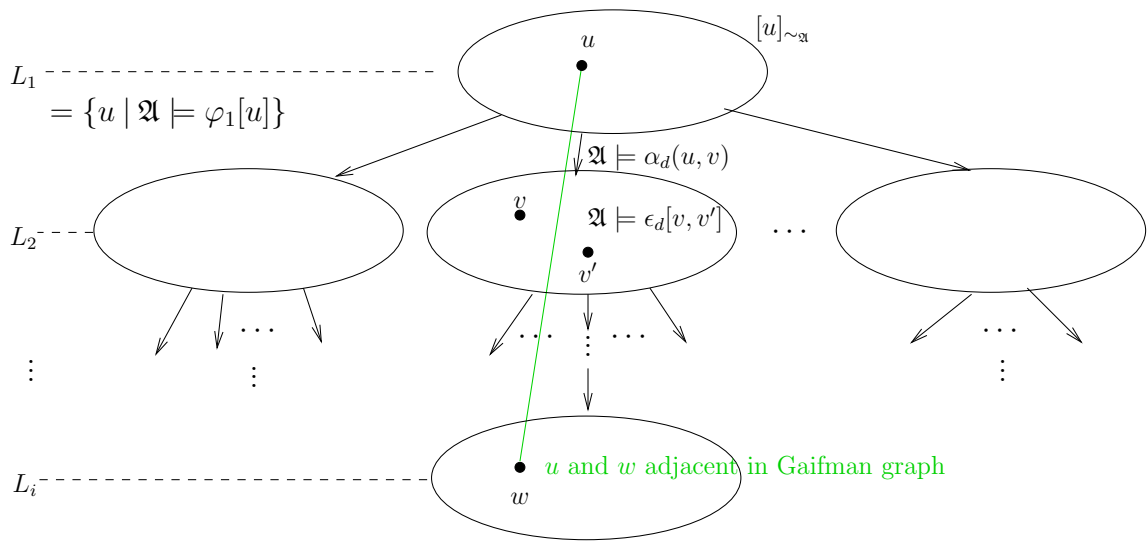}
\end{center}
\caption{The canonical tree-decomposition defined in \FO.}
\label{fig:canondecomp}
\end{figure}

\section{Conclusion}
\label{sec:conclusion}

We have investigated the expressive power and the relative
succinctness of different classes of logical formulae on structures of
bounded tree-depth $d$.  In particular, we have shown that, if a class
$\class C$ of such structures is $\MSO{}$-definable or
order-invariantly $\FO{}$-definable, then it is also
$\FO{}$-definable.  For $\MSO{}$-definable classes, this was already
known. But, in both cases, our approach also shows that the size of
the $\FO{}$-sentence which defines $\class C$ is at most $d$-fold
exponential in the quantifier-rank of a given order-invariant $\FO{}$-
or $\MSO{}$-sentence which defines $\class C$.  For $\MSO{}$-formulae,
we have proved that this upper bound on the size of the
$\FO{}$-sentence is essentially optimal.  It would be interesting to
know if there is a corresponding lower bound for the result about
order-invariantly $\FO{}$-definable classes.

One motivation to consider bounded tree-depth graphs was the role of
these graphs in the theory of sparse graphs which has been outlined in
the book \cite{NesetrilMendez2012}. This link has been exploited in
several results about the algorithmic behaviour of logics on sparse
structures. Can our results on order-invariant $\FO{}$-sentences on
bounded tree-depth structures be used to obtain results about such
sentences on more general classes of sparse structures?

An interesting extension of order-invariance is
\emph{addition-invariance} where sentences are not only allowed to use
some linear order but also the graph of the addition operation that is
induced by the embedding of a structure into the natural numbers that
comes with the linear order.  The paper \cite{Schweikardt2010}
obtained a characterisation of the classes of structures which are
addition-invariantly $\FO{}$-definable over unary signatures, i.e. on
structures of tree-depth $1$.  Each such class of structures is
definable in $\FO{}_{\textup{card}}$, i.e. the extension of $\FO{}$
with nullary predicates $C_{m}$, for all positive integers $m$, which
state that the cardinality of a structure is divisible by $m$.  Our
proofs hinge on the composition method and there is no obvious way how
these methods could be extended to addition-invariant formulae.  Does
addition-invariant $\FO{}$ have the same expressive power as
$\FO{}_{\textup{card}}$ on bounded tree-depth structures?

\
\paragraph{Acknowledgements} 

We want to thank Isolde Adler for bringing the first two authors together with
the third author, and Nicole Schweikardt for her helpful suggestions.

\bibliographystyle{plain}
\bibliography{paper}

\begin{thebibliography}{10}

\bibitem{BenediktSegoufin2009}
Michael~A. Benedikt and Luc Segoufin.
\newblock Towards a characterization of order-invariant queries over tame
  graphs.
\newblock {\em Journal of Symbolic Logic}, 74(1):pp. 168--186, 2009.

\bibitem{BoulandDK12}
Adam Bouland, Anuj Dawar, and Eryk Kopczynski.
\newblock On tractable parameterizations of graph isomorphism.
\newblock In {\em Proc. IPEC 2012}, pages 218--230, 2012.

\bibitem{ChandraH1982}
Ashok Chandra and David Harel.
\newblock Structure and complexity of relational queries.
\newblock {\em {JCSS}}, 25(1):pp. 99--128, 1982.

\bibitem{Courcelle1991}
Bruno Courcelle.
\newblock The monadic second-order logic of graphs v: on closing the gap
  between definability and recognizability.
\newblock {\em Theoretical Computer Science}, 80:pp.~153--202, 1991.

\bibitem{Courcelle1996}
Bruno Courcelle.
\newblock The monadic second-order logic of graphs x: linear orderings.
\newblock {\em Theoretical Computer Science}, 160(1--2):pp.~87--143, 1996.

\bibitem{CourcelleE2012}
Bruno Courcelle and Joost Engelfriet.
\newblock {\em Graph Structure and Monadic Second-Order Logic -- A
  Language-Theoretic Approach}.
\newblock Cambridge University Press, 2012.

\bibitem{EickmeyerEH2014}
Kord Eickmeyer, Michael Elberfeld, and Frederik Harwath.
\newblock Expressivity and succinctness of order-invariant logics on
  depth-bounded structures.
\newblock In {\em Proceedings of the 39th International Symposium on
  Mathematical Foundations of Computer Science (MFCS 2014), Part {I}}, pages
  256--266, 2014.

\bibitem{ElberfeldGT12}
Michael Elberfeld, Martin Grohe, and Till Tantau.
\newblock Where first-order and monadic second-order logic coincide.
\newblock In {\em Proc. LICS 2012}, pages 265--274. IEEE Computer Society,
  2012.

\bibitem{FlumGrohe2004}
Jörg Flum and Martin Grohe.
\newblock {\em Parameterized Complexity Theory}.
\newblock Springer-Verlag, 2006.

\bibitem{GajarskyH2012}
Jakub Gajarsk\'y and Petr Hlin\v{e}n\'{y}.
\newblock Faster deciding {MSO} properties of trees of fixed height, and some
  consequences.
\newblock In {\em Proc. FSTTCS 2012}, pages 112--123, 2012.

\bibitem{GroheSchweikardt05}
Martin Grohe and Nicole Schweikardt.
\newblock {The succinctness of first-order logic on linear orders}.
\newblock {\em Logical Methods in Computer Science}, 1(1:6):pp. 1--25, 2005.

\bibitem{Howie1976}
J.M.Howie.
\newblock {\em An introduction to semigroup theory}.
\newblock Academic Press, 1976.

\bibitem{Libkin2004}
Leonid Libkin.
\newblock {\em Elements of Finite Model Theory}.
\newblock Springer-Verlag, 2004.

\bibitem{Makowsky2004}
J.A. Makowsky.
\newblock Algorithmic uses of the {Feferman--Vaught} theorem.
\newblock {\em Annals of Pure and Applied Logic}, 126(1--3):pp. 159--213, 2004.

\bibitem{Minsky1967}
Marvin~L. Minsky.
\newblock {\em Computation: Finite and Infinite Machines}.
\newblock Prentice-Hall, Inc., Upper Saddle River, NJ, USA, 1967.

\bibitem{NesetrilMendez2012}
Jaroslav Ne\v{s}et\v{r}il and Patrice Ossona~de Mendez.
\newblock {\em Sparsity: Graphs, Structures, and Algorithms}.
\newblock Springer-Verlag Berlin Heidelberg, 2012.

\bibitem{Schweikardt2010}
N.~Schweikardt and L.~Segoufin.
\newblock Addition-invariant {FO} and regularity.
\newblock In {\em Proc.\ 25th IEEE Symposium on Logic in Computer Science
  (LICS'10)}, pages 285--294. IEEE, 2010.

\bibitem{Schweikardt2013}
Nicole Schweikardt.
\newblock A short tutorial on order-invariant first-order logic.
\newblock In {\em Proc. CSR 2013}, pages 112--126, 2013.

\end{thebibliography}

\end{document}